\RequirePackage{afterpackage}
\AfterPackage{amsthm}{\RequirePackage{thmtools}}
\AfterPackage{amsmath}{\RequirePackage[shortlabels]{enumitem}}
\documentclass[USenglish,thm-restate,cleveref]{lipics-v2021}
\usepackage{amsmath,amsthm,mathtools}
\usepackage{multicol}

\usepackage{fix-cm} 

\usepackage{amssymb,xspace}
\usepackage{soul}
\usepackage{microtype}

\usepackage{multicol}
\usepackage{datetime}
\usepackage{tikz}
\usetikzlibrary{automata,arrows.meta,tikzmark,calc,backgrounds}
\usetikzlibrary{decorations.pathreplacing}
\usepackage{comment}
\usepackage{etoolbox}
\newbool{for_arxiv}
\booltrue{for_arxiv} 

\input{define.orgtex}

\makeatletter
\providecommand*{\shuffle}{%
  \mathbin{\mathpalette\shuffle@{}}%
}
\newcommand*{\shuffle@}[2]{%
  \sbox0{$#1\vcenter{}$}%
  \kern .15\ht0 
  \rlap{\vrule height .25\ht0 depth 0pt width 2.5\ht0}%
  \raise.1\ht0\hbox to 2.5\ht0{%
    \vrule height 1.75\ht0 depth -.1\ht0 width .17\ht0 %
    \hfill
    \vrule height 1.75\ht0 depth -.1\ht0 width .17\ht0 %
    \hfill
    \vrule height 1.75\ht0 depth -.1\ht0 width .17\ht0 %
  }%
  \kern .15\ht0 
}
\newcommand\labelandarrowsSilent[1]{%
  \ifthmt@thisistheone%
   \label{#1}%
  \else
   \@label{starred:#1}%
  \fi%
  \ignorespaces
}
\makeatother

\nolinenumbers
\ifbool{for_arxiv}{
  \hideLIPIcs
}{}

\title{Shuffles of Context-Free Languages along Regular Trajectories}

\author{Corentin Barloy}{University of Bochum, Germany}{corentin.barloy@rub.de}{0000-0001-5420-8761}{is supported by the Deutsche
Forschungsgemeinschaft (DFG, German Research Foundation), grant 532727578}
\author{Michaël Cadilhac}{DePaul University, USA}{michael@cadilhac.name}{0000-0001-9828-9129}{}
\author{Kyle Ockerlund}{Google, USA}{ockerlundkyle@google.com}{0009-0004-3496-7137}{}

\authorrunning{C. Barloy, M. Cadilhac, K. Ockerlund}

\Copyright{Corentin Barloy, Michaël Cadilhac, and Kyle Ockerlund} 


\acknowledgements{\Cref{prop:nosimplecflp} is due to Georg Zetzsche.  We
thank Arka Ghosh and Petra Wolf for stimulating discussions on this topic.}
\ifbool{for_arxiv}{
}{
  \relatedversiondetails{Full Version}{https://arxiv.org/pdf/2603.26162}
  \let\labelandarrows\labelandarrowsSilent
}

\def\jas{Jan\v{c}ar and \v{S}\'ima\xspace}
\defmathtext\cfl{CFL}
\defmathtext\dcfl{DCFL}
\defmathtext\cflp{CFL\('\)}
\defmathtext\dcflp{DCFL\('\)}
\defmath\stk{\mathsf{stk}}
\defmath\stt{\mathsf{stt}}
\def\ang#1{\langle #1\rangle}
\def\ange#1{\ang{#1}^\eps}
\defmath\eps\varepsilon
\defmath\hard{\mathsf{hard}}
\defmath\easy{\mathsf{easy}}
\defmath\pkh\Psi
\defmath\core{\mathsf{core}}
\defmath\pc{\mathsf{pc}}
\defmath\st{\mathsf{st}}
\defmath\fold{\mathsf{fold}}
\defmath\excl{\mathsf{excl}}
\defmath\shuf\shuffle
\defmath\len{\mathsf{lengths}} 


\newcommand{\mysubparagraph}[1]{\vspace{.15cm}\noindent \textbf{\textsf{#1}} }

\ifbool{for_arxiv}{}{
\EventEditors{Sayan Bhattacharya, Danupon Nanongkai, Michael Benedikt, and Gabriele Puppis}
\EventNoEds{4}
\EventLongTitle{53rd International Colloquium on Automata, Languages, and Programming (ICALP 2026)}
\EventShortTitle{ICALP 2026}
\EventAcronym{ICALP}
\EventYear{2026}
\EventDate{July 7--10, 2026}
\EventLocation{Royal Holloway, University of London, Egham, United Kingdom}
\EventLogo{}
\SeriesVolume{374}
\ArticleNo{178}
\category{Track B: Automata, Logic, Semantics, and Theory of Programming}
}

\begin{document}

\maketitle

\begin{abstract}
  In single-core processors, concurrency requires that multiple processes
  be interleaved into a single thread of execution by a scheduler. The
  language-theoretic operation that corresponds to this is the \emph{shuffle} of
  two languages: the set of words obtained by interleaving a word from each
  language in an arbitrary, letter-wise fashion.  It is well known that regular
  languages are closed under shuffles, while context-free languages (CFLs) are
  not.

  Following an established line of research, this paper considers shuffles
  according to regular ``trajectories,'' that is, subject to scheduling
  constraints expressed by an automaton.  Unsurprisingly, some trajectories
  allow for CFLs to be shuffled into CFLs (e.g., simple concatenation of the two
  words), while others do not. This paper provides a robust toolset to show that
  a given trajectory would \emph{always} shuffle two nonregular CFLs into a
  nonCFL.  In the case of deterministic CFLs (DCFLs), a salient trichotomy of
  trajectories depending on how they shuffle DCFLs is provided.

  These results are based on lemmata of independent interest regarding how
  pushdown automata (PDA) must invoke the stack when accepting a nonregular CFL
  or DCFL. The latter case relies on a recent result of \jas (MFCS'2021);
  answering an open question therein, it is demonstrated that said result cannot
  be generalized to arbitrary CFLs, leading to dedicated machinery for both
  cases.
\end{abstract}

\clearpage
\ifbool{for_arxiv}{
\clearpage
\setcounter{page}{0}
\noindent
To facilitate the reading, we provide a table of contents
below. We further provide hyperlinks for navigating between
statements in the main text and proofs in the appendix. To go to the
proof of a statement, click ``\faArrowDown'' in the left margin, and then
click ``\faArrowUp'' to go back to the statement in the main text.
\tableofcontents
\clearpage
}{}

\section{Introduction}

In the theory of concurrency and parallel computation, many composition
operations ultimately boil down to a variant of the shuffle operation (see,
e.g., \cite{OgdenRR78,Shaw78,nivat1982,iwama83,bergstrak84,baeten90,salomaay99}
for early examples from the 1970s to the 1990s).  The theoretical study of
shuffles was thus initiated to formalize the interactions and
synchronization of processes.

\mysubparagraph{Types of shuffles.} As concurrency may take many forms,
theoretical work has turned to restricting the ``full'' shuffle operation to
more constrained interleavings between processes. For instance,
Bérard~\cite{berard87} introduced \emph{literal shuffles} which enforce that the
processes strictly alternate --- she refers to it as modeling step-lock
transmission in single communication channels.  Further widening the definition
of shuffles, Mateescu et al.~\cite{mateescurs98} introduced \emph{shuffles along
a trajectory}, where a trajectory is a language of words over \(\{s, t\}\), and
such a word indicates in which order the processes are scheduled (\(s\) for the
first one, \(t\) for the second).  For example, \((s+t)^*\) is the trajectory
of the full shuffle, and \((st)^*\) would force that the two processes strictly
alternate over the same length of time.  Shuffles along trajectories generalize
many language-theoretic operations, including for example concatenation
(trajectory \(s^*t^*\)).

\mysubparagraph{Questions about shuffles.}  Historically, the most prominent
questions about shuffles revolved around closure properties.  While regular
languages are easily seen to be closed under the full shuffle, Latteux famously
showed~\cite{latteux79}, in the late 1970s, that the full shuffle of two CFLs is
a CFL if, and only if, one of the two shuffled languages is
regular.\footnote{Let us note here that this only holds if the two languages are
over disjoint alphabets, which is a standard assumption we will be making.
Shuffling languages over a common alphabet is, in addition to being less natural
in most concurrent scenarios, not well behaved: one can show that for any
language \(E \subseteq \Sigma^*\) containing the empty word $\eps$, with \(\#
\notin \Sigma\), the full shuffle of \(\{w_1\# w_2 \mid w_1 \in \Sigma^* \land
w_2 \in E\}\) and \(\{w_1\# w_2 \mid w_1 \in E \land w_2 \in \Sigma^*\}\) is
simply \(\Sigma^*\#\Sigma^*\#\Sigma^*\).} %

In fact, Ogden et al.~\cite{OgdenRR78} exhibited DCFLs that shuffle into
NP-complete languages.  On the other hand, Bérard~\cite{berard87} shows that
there are nonregular CFLs which can be shuffled along the trajectory
\((st)^*(s^*+t^*)\) into a CFL.  She asks in particular if this is also true of
the trajectory \((s^*+t^*)(st)^*(s^*+t^*)\), a question that has been open since
then.  Mateescu et al.~\cite{mateescurs98} expand on these results, focusing in
particular on regular trajectories (Section 5 therein).  They show, for
instance, that the shuffle of a CFL and a regular language along a regular
trajectory is always CFL.

This fragmented landscape leaves open an important question: which regular
trajectories \emph{do} preserve context-freeness, in the sense that the shuffle
of two CFLs is a CFL?
These questions can be seen as variants of --- or a preliminary step to --- the
so-called \emph{parallelization problem}, which asks, given a nonCFL, whether it
is expressible as the shuffle along some trajectory of two CFLs.  Crisp answers
to these questions would not only clarify the landscape of regular trajectories,
but allow for an understanding of when regular schedulers over context-free
processes require additional computational power to be seen as a single thread
of context-free execution.

\mysubparagraph{Contributions.} The contributions of this paper follow two broad
lines:
\begin{enumerate}[noitemsep,topsep=0pt,parsep=0pt,partopsep=0pt]
\item In \Cref{sec:coup}, we investigate 'couplings' within PDA. Briefly, two
  subsections of a word are said to be \emph{coupled} by a PDA if, for all accepting
  runs, a symbol is pushed during the first subsection, and popped during the
  second.  We show, relying on some surgery on the runs, that if couplings are
  rare, then the PDA recognizes a regular language.  A strengthening is provided
  for the case where the PDA, which may itself be nondeterministic, recognizes a
  deterministic CFL, relying on a deep result of \jas~\cite{jancars21}.  The
  proofs of these lemmas, dubbed the Coupling Lemmas, are provided in
  \Cref{sec:couppda,sec:coupdpda}, which can be safely skipped on first reading.
  We note that these sections are entirely independent of shuffles; they simply
  provide a toolbox to study how couplings appear.
\item In \Cref{sec:shuf}, we turn to shufflings proper.  We first study how
  couplings survive a shuffle, leading to two so-called Resilience Lemmas.  We
  then show two primary results: (1) We provide infinite families of
  trajectories that \emph{always} shuffle two nonregular CFLs into a nonCFL; and
  (2) We partition regular trajectories into three categories:
  \emph{\cflp_safe}, under which shuffles of CFLs are always CFLs;
  \emph{\dcflp_hostile}, under which shuffles of nonregular DCFLs are always
  nonCFL; \emph{\dcflp_mixed}, under which nonregular DCFLs shuffle sometimes
  into CFLs, sometimes into nonCFLs.  We give a decidable characterization of
  each of these categories.
\end{enumerate}

We note that we partially answer Bérard's open question: we can show that the
shuffle along \((s^*+t^*)(st)^*(s^*+t^*)\) \emph{never} preserves nonregular
DCFLs, but we cannot answer the same for general CFLs.  However, our results
also imply that the same is true for the trajectory \((st)^*(s^*+t^*)\), for
which Bérard showed there \emph{are} nonregular CFLs that shuffle into a CFL,
hinting at the fact that a full classification for CFL-preservation might prove
to be much more complex.

\section{Preliminaries}
\label{sec:prelim}

We rely on standard concepts of automata theory, as presented for instance by
Sipser~\cite{sipser97}.  We make precise only the notations that are somewhat
idiosyncratic, but expect the reader to know about languages, finite
automata, pushdown automata, and context-free languages.

\mysubparagraph{Words.}  We usually write \(\Sigma\) for alphabets and \(\eps\) for the
empty word.  For a word \(w \in \Sigma^*\), we write \(|w|\) for its length,
\(w[i]\) for the \(i\)-th letter of \(w\), \(w[i:j]\) for the subword from \(i\) to
\(j\), and abbreviate \(w[1:i]\) as \(w[:i]\) and \(w[i:|w|]\) as \(w[i:]\).

\mysubparagraph{PDA, CFL, DCFL.}  A PDA is a tuple
\(A = (Q, \Sigma, \Gamma, \delta, q_0, F)\) with \(Q\) a set of states,
\(\Sigma\) the input alphabet, \(\Gamma\) the stack alphabet,
\(\delta \colon Q \times (\Sigma \cup \{\eps\}) \times (\Gamma \cup \{\eps\}) \to \mathcal{P}(Q \times (\Gamma \cup \{\eps\}))\) the
transition function, \(q_0\) the initial state, and \(F\) the set of final states.
The condition \(\delta(q, a, s) \ni (q', s')\) is to be interpreted as: from state \(q\), reading
\(a\) with \(s\) on top of the stack, go to state \(q'\) and replace the top of the
stack with \(s'\).  Note that any of \(a, s, s'\) can be \(\eps\), which corresponds to not
reading from the word, not reading from the stack, or not writing to the stack,
respectively.  When a (pushdown) automaton is under discussion, we will use the
elements of the tuple as above freely, without repeating ``\(A = (Q, \Sigma,
\Gamma, \delta, q_0, F)\).'' The PDA is deterministic if for every state \(q\),
letter \(a \in \Sigma\) and top of stack \(s \in \Gamma\), there is only one
valid transition that can be taken.

PDA define the class of context-free languages (CFLs) while deterministic PDA
define the class of deterministic context-free languages (DCFLs).  Following
\jas~\cite{jancars21}, we write \cflp for the set of CFLs that are not regular,
and similarly for \dcflp.

A \emph{run} \(r\) is a finite sequence of transitions in the PDA, such that
\(r[i]\) ends where \(r[i+1]\) starts for every \(i\) and the stack semantic is
correct.  The label \(w\) of \(r\) is the concatenation of the labels of its
transitions.  The content of the stack at the end of \(r\) is written
\(\stk(r) \in \Gamma^*\), with the top of the stack at the rightmost position, and we
abbreviate \(\stk(r[:i])\) as \(\stk(r, i)\).  Similarly, \(\stt(r)\) is the state
reached by \(r\), and \(\stt(r[:i])\) is abbreviated as \(\stt(r,i)\).  A run is
\emph{accepting} if it starts in \(q_0\) with an empty stack and ends in a state
in \(F\) (the contents of the stack not being relevant for acceptance).

\mysubparagraph{From runs to labels.} Since \(r\) may contain
\(\eps\)_transitions, we introduce notations for the correspondence between the
letters of its label \(w\) and the transitions of \(r\): we let \(r\ang{i}\) be the
transition that reads \(w[i]\).  We let \(r\ang{i:j}\) be the transitions that read
the subword \(w[i:j]\), excluding \eps_transitions at the extremities --- that is,
\(r\ang{i:j}\) starts with \(r\ang{i}\) and ends with \(r\ang{j}\).  We let
\(r\ange{i:j}\) be similar to \(r\ang{i:j}\) but including \eps_transitions at the
extremities --- it still holds that the label of the subrun \(r\ange{i:j}\) is
\(w[i:j]\).  Finally, when \(w = xyz\), where \(y = w[i:j]\) for some \(i, j\), we
simply write, abusing notation, \(r\ang{y}\) for \(r\ang{i:j}\) and similarly for
\(r\ange{y}\).  In all such cases, the particular factorization of \(w\) in which
\(y\) appears is provided to avoid potential ambiguity arising from multiple
occurrences of \(y\).  Note, in all of these notations, we see the transitions
extracted from \(r\) \emph{within the context} of \(r\); for instance,
\(r\ang{i:j}\) would formally be defined as the set of \emph{positions} in \(r\)
of the transitions that read \(w[i:j]\).  We sometimes make use of this, and
allow writing, e.g., \(k \in r\ang{y}\) to mean that \(r[k]\) is a transition in
the subrun of \(r\) that reads \(y\).

\section{Couplings and the Coupling Lemmas}\label{sec:coup}

A general PDA (accepting a deterministic CFL or not) can utilize the stack in
many diverse ways; we focus our investigation on the following property.

\begin{plaindef}[Couplings]
  Let \(A\) be a PDA and \(r\) a run.  We say that two positions \(i\) and \(j\)
  with \(i < j\) are \(r\)_\emph{coupled} if \(\stk(r, i-1) = \stk(r, j)\)
  and no \(k\) with \(i \leq k < j\) has a shorter \(\stk(r, k)\). That is,
  a symbol is pushed by \(r[i]\) and popped by \(r[j]\). For \(w = uxvyz\) the
  label of \(r\), we say that \(x\) and \(y\) are \(r\)_\emph{coupled} if there
  are positions \(i \in r\langle x\rangle, j \in r\langle y\rangle\) that are
  coupled in \(r\).  Finally, for \(w = uxvyz\) any word, we say that \(x\) and
  \(y\) are \(A\)_\emph{coupled} if they are \(r\)_coupled in all
  \emph{accepting} runs \(r\) of \(A\) on \(w\).
\end{plaindef}

\begin{example}\label{ex:def} Consider the prototypical language \(L = \{a^n b^n
\mid n \geq 0\}\). Clearly, any PDA which accepts \(L\) must push to the stack
while reading \(a^n\) and pop while reading \(b^n\); otherwise, they could be
pumped independently to leave the language. Hence, for each PDA \(A\) accepting
\(L\) there is some value $q$ such that, for all \(n \geq q\), \(a^n\) and
\(b^n\) are \(A\)-coupled in the word \(a^n b^n\).
\end{example}

\begin{remark}
  In our applications, we will identify, in proofs by contradiction, shufflings
  that would have couplings that \emph{cross} each other, e.g., \(x_1x_2y_1y_2\) such that
  \(x_1\) and \(y_1\) are \(A\)_coupled, and so are \(x_2\) and \(y_2\).  This obviously
  violates stack discipline, and leads to the desired contradictions.  We also
  stress that \(A\)_coupling is a property that ranges only over \emph{accepting}
  runs; there may be nonaccepting runs in which \(x\) and \(y\) are uncoupled.
\end{remark}

\begin{lemma}[Coupling Lemma]
  Let \(L\) be a CFL recognized by a PDA \(A\).  If there are
  \(k_1, k_2, k_3 \in \bbN\) such that for every \(w = uxvyz \in L\) with
  \(|u|=k_1, |v|=k_2, |z|=k_3\), we have that \(x\) is \emph{not} \(A\)_coupled with
  \(y\), then \(L\) is regular.
\end{lemma}

\begin{remark}
  The converse does not hold.  Consider the standard 3-state DPDA \(A\) for
  \(\{a^nb^n \mid n \geq 0\}\) and make every state accepting, including the
  sink state.  For any \(k_1, k_2, k_3 > 0\), let \(u = a^{k_1}, v = a^{k_2}, z
  = b^{k_3}\), set \(x= a^{k_3}, y = b^{k_1+k_2}\) and consider the word \(w =
  uxvyz\).  The only run of \(A\) on \(w\) will push a symbol in \(x\) that will
  be popped in \(y\), so \(x\) and \(y\) are \(A\)_coupled.
\end{remark}

\begin{lemma}[Deterministic Coupling Lemma]\label{lem:detcoup}
  Let \(L\) be a DCFL recognized by a (not necessarily deterministic) PDA \(A\).  If
  for all sublanguages \(\{ux^nvy^nz \mid n \geq 0\} \subseteq L\) there is an
  \(n\) and \(1 \leq i \leq n\) such that, in \(ux^nvy^nz\), the \(i\)-th occurrence of
  \(x\) is \emph{not} \(A\)_coupled with the \((n-i+1)\)-th occurrence of \(y\), then
  \(L\) is regular.
\end{lemma}

\begin{example} Couplings in PDA for general CFL can have much less structure
than \cref{lem:detcoup} requires for those recognizing DCFL, precluding most
standard CFL machinery. For example, consider the Goldstine language as first
applied to shuffles by Bérard~\cite{berard87}: \(G = \{a^{n_1}b \ldots
a^{n_i}b \ldots a^{n_k}b \mid k \geq 1, n_i \geq 0, \exists i \text{ with } i \neq n_i\)\}.
Multiple PDA reading \(G\) exist. For instance, for a nondeterministically
chosen $i$, one PDA might compare adjacent blocks of $a$'s for the criteria that
$n_i \neq n_{i-1} + 1$, while another PDA may choose to check
$n_i$ against the preceding $b$'s for the condition that $n_i \neq i$. In
either case: pumping on $G$ may rapidly introduce new $i$ which satisfy these
conditions, each of which would get a dedicated run, breaking any pre-existing
couplings.
\end{example}

\section{Proof of the Coupling Lemma}\label{sec:couppda}

Broadly, the goal of this proof is to demonstrate that if certain couplings
do not exist in some PDA, then it becomes possible to capture the behavior of
arbitrarily long strings in a finite number of classes towards regularity. Our
main object to encapsulate this is the \emph{endcap,} presented below.

\subsection{Definitions: endcaps}

Let \(A\) be a PDA and \(uxv \in \Sigma^*\) be the label of some run \(r\).  The
\emph{endcap} at positions \(i \in r\ange{u}\) and \(j \in r\ange{v}\) is a tuple that
contains information about the subruns \(r[:i]\) and \(r[j:]\), that is, skipping
the portion of \(r\) between \(i\) and \(j\).  Informally, the tuple contains the
portion of \(u\) read until \(r[i]\), the states reached by \(r[i]\) and
\(r[j]\), and the portion of \(v\) read from \(r[j]\).  Formally, the endcap is the
tuple \((u[:c_i], \stt(r, i), \stt(r, j), v[c_j:])\), where \(c_i\) is the smallest
such that \(i \in r\ange{c_i}\) and \(c_j\) the largest such that \(j \in r\ange{c_j}\).

In this setting, we will always focus on a specific endcap, that is unique for
each decomposition \(uxv \in \Sigma^*\) of the label of the run \(r\). Let
\(h_x\) be the smallest stack-height observed while reading \(x\) with
\(r\ang{x}\) and \(h_v\) the smallest stack-height observed while reading \(v\) with
\(r\ange{v}\) (note that the minimum may be achieved potentially \emph{right
  before} the first transition; we allow that as an edge case).  We let \(i\) be
the rightmost position in \(r\ange{u}\) that has maximal stack-height below
\(h_x\) and \(h_v\).  We let \(j\) be the leftmost position in \(r\ange{v}\) with
stack-height \(h_v\) (again, as an edge case, \(j\) may be the last transition
before \(\min r\ange{v}\)).  A graphical representation appears in
\Cref{fig:minendcap}.  The endcap at these positions \(i\) and \(j\) is called the
\emph{min-endcap} of \(r\) with respect to the decomposition \(uxv\) of its label.
Note, our definition of PDA requires that at most one element is pushed/popped
per transition.  Hence, the selected transitions are precise such that we need
not worry about skipping any stack heights in a single transition.

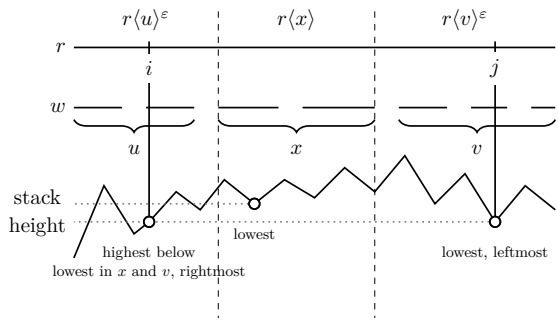
\begin{figure}[!h]
  \centering
  \scalebox{0.8}{\begin{tikzpicture}[>=stealth,thick]

\draw[dashed,thin] (2.4,0.6) -- (2.4,-4.5);
\draw[dashed,thin] (5,0.6)   -- (5,-4.5);

\node[left] at (0,0) {$r$};
\draw (0,0) -- (8,0);

\node[above] at (1.2,0.15) {$r\ange{u}$};
\node[above] at (3.7,0.15) {$r\ang{x}$};
\node[above] at (6.5,0.15) {$r\ange{v}$};

\begin{scope}[yshift=-1cm]
  \node[left] at (0,0) {$w$};

  \draw (0,0) -- (0.8,0);
  \draw (1.2,0) -- (2,0);

  \draw (2.4,0) -- (3.4,0);
  \draw (3.8,0) -- (5,0);

  \draw (5.4,0) -- (6.2,0);
  \draw (6.6,0) -- (7.2,0);
  \draw (7.6,0) -- (8,0);

  \draw[decorate,decoration={brace,mirror,amplitude=6pt}]
    (0,-0.2) -- (2,-0.2) node[midway,yshift=-0.5cm] {$u$};
  \draw[decorate,decoration={brace,mirror,amplitude=6pt}]
    (2.4,-0.2) -- (5,-0.2) node[midway,yshift=-0.5cm] {$x$};
  \draw[decorate,decoration={brace,mirror,amplitude=6pt}]
    (5.4,-0.2) -- (8,-0.2) node[midway,yshift=-0.5cm] {$v$};
\end{scope}

\begin{scope}[yshift=-3.5cm]
  \coordinate (stack_center) at (0,0.75);
  \node[left,align=center,font=\normalsize] at (0,0 |- stack_center) {stack\\height};

  \draw[thick]
    (0,0)           -- (0.5,1.2)  coordinate (s05)
    -- (1.0,0.4)     coordinate (s10)
    -- (1.25,0.6)    coordinate (s12)   
    -- (1.7,1.1)     coordinate (s16)
    -- (2.1,0.8)     coordinate (s20)
    -- (2.5,1.3)     coordinate (s25)
    -- (3.0,0.9)     coordinate (sxmin) 
    -- (3.5,1.3)     coordinate (s35)
    -- (4.0,1.0)     coordinate (s40)
    -- (4.5,1.5)     coordinate (s45)
    -- (5.0,1.1)     coordinate (s50)
    -- (5.5,1.7)     coordinate (s55)
    -- (6.0,0.9)     coordinate (s60)
    -- (6.5,1.4)     coordinate (s65)
    -- (7.0,0.6)     coordinate (svmin) 
    -- (7.5,1.2)     coordinate (s75)
    -- (8.0,0.8)     coordinate (s80)
  ;

  \draw[dotted,black!50] (sxmin) -- (0,0.9); 
  \draw[dotted,black!50] (svmin) -- (0,0.6);
\end{scope}

\draw[thick] (s12) -- ++(0,3); 
\node[fill=white] at ($(s12)+(0,2.55)$) {$i$};

\draw[thick] (svmin) -- ++(0,3); 
\node[fill=white] at ($(svmin)+(0,2.55)$) {$j$};

\filldraw[fill=white] (sxmin) circle (2.5pt); \draw (sxmin) circle (2.5pt);
\filldraw[fill=white] (svmin) circle (2.5pt); \draw (svmin) circle (2.5pt);
\filldraw[fill=white] (s12) circle (2.5pt); \draw (s12) circle (2.5pt);

\node[below=3mm,align=center,scale=0.8] at (sxmin) {\small lowest};
\node[below=3mm,scale=0.8] at (svmin) {\small lowest, leftmost};
\node[below=3mm,align=center,scale=0.8,font=\small] at (s12) {highest
  below\\ lowest in \(x\) and \(v\), rightmost};

\end{tikzpicture}}
  \caption{Min-endcaps: how to find \(i\) and \(j\).  The \(r\) line denotes a
    sequence of transitions.  The \(w\) line a sequence of letters; the blanks in
    that line correspond to \eps_transitions.}
  \label{fig:minendcap}
\end{figure}

\begin{restatable}{fact}{FactEndcap}\labelandarrows{fact:endcap}
  Let \(uxv \in \Sigma^*\) be the label of some run \(r\) that has a min-endcap from \(i\) to
  \(j\).  For all \(i' \geq i\), \(|\stk(r, i')| \geq |\stk(r, i)|\).  For all \(j' \geq j\),
  \(|\stk(r,j')| \geq |\stk(r, j)|\).
  
  In addition, we have that \(\stk(r, j) = \stk(r, i)e\) for some
  \(e \in \Gamma^*\) and all symbols of \(e\) are pushed by the transitions
  \(r\ang{x}\).  We call \(e\) the \emph{stack-effect} of the min-endcap.
\end{restatable}

\subsection{Surgery on runs}

\begin{plaindef}
  Let \(A\) be a PDA.  For \(s \in \Gamma^*\) and \(q \in Q\), write \(A_{s,q}\) for the PDA \(A\)
  in which the stack is pre-initialized with \(s\) and the initial state set to
  \(q\).  Two sequences \(e, e' \in \Gamma^*\) are \(k\)_suffix-equivalent if for any \(s \in
  \Gamma^*\) and \(q \in Q\), the PDA \(A_{se, q}\) and \(A_{se', q}\) accept the same
  words of length at most \(k\).  (Recall that \(se\) means that \(e\) is pushed
  on top of \(s\).)
\end{plaindef}

\begin{restatable}{lemma}{LemFinClass}\labelandarrows{lem:finclass}
  Let \(A\) be a PDA.  For any fixed \(k\), \(k\)_suffix-equivalence is an equivalence
  relation with a finite number of classes, denoted \([e]_k\) for any
  representative \(e \in \Gamma^*\).  For convenience, we let \(\eps\) be in its own
  equivalence class.
\end{restatable}

\begin{restatable}{lemma}{LemSurgery}\labelandarrows{lem:surgery}
  Let \(A\) be a PDA, let \(r = r_1r_y\) be a run of \(A\) and write the label of
  \(r_1\) as \(w_1 = u_1x_1v_1\) and the label of \(r_y\) as \(y\).  Assume that
  \(x_1\) and \(y\) are not coupled in \(r\).  Consider another run \(r_2\) with label
  \(w_2 = u_2x_2v_2\) that has the same min-endcap as \(r_1\).  If the stack-effects
  of the min-endcap in \(r_1\) and \(r_2\) are both empty, or both nonempty and
  \(k\)-suffix equivalent, then for any word \(z \in \Sigma^{\leq k}\), if
  \(w_1yz\) is accepted by a run starting with \(r_1r_y\), then \(w_2yz\) is accepted
  by \(A\).
\end{restatable}

\subsection{Proof of the Coupling Lemma}

Recall the assumptions of the lemma: \(L\) is a CFL recognized by a PDA \(A\) and
there are \(k_1, k_2, k_3 \in \bbN\) such that for every \(w = uxvyz \in L\) with
\(|u|=k_1, |v|=k_2, |z|=k_3\), we have that \(x\) is \emph{not} \(A\)_coupled with
\(y\).

\mysubparagraph{Equivalence classes.} We associate every word \(w \in \Sigma^*\) with a
finite object and show that the induced equivalence relation refines the
Myhill-Nerode equivalence, that is, the relation \(x \equiv_L y\) that holds if for any
\(z\), \(xz \in L\) iff \(yz \in L\).  This relation is of finite index iff
\(L\) is regular, and we now construct an equivalence relation of finite index
refining \(\equiv_L\), showing that \(L\) is regular.

If \(|w| \leq k_1 + k_2\), we let \(w\) be alone in its equivalence class.  Otherwise,
define \(C\) to be the set of min-endcaps for \emph{any} way of writing
\(w = uxv\) and any accepting run over \(uxvyz\) with \(|z| = k_3\) and \(x\) and
\(y\) not coupled. We then associate \(w\) with the pair \((M, E)\) defined as:
\begin{align*}
  M & = \{ s \in \Sigma^{< k_3} \mid w\cdot s \in L\}\\
  E & = \{ (c, [e]_{k_3}) \mid c \in C \text{ with
      stack-effect } e\}
\end{align*}
If \(w\) has no accepting run, then this pair is empty. Importantly, these are
finite objects per \Cref{lem:finclass}.

\mysubparagraph{Refinement of Myhill-Nerode.}  Let \(w_1, w_2\) be two words that
are equivalent by the previously defined relation.  We show that for any
\(s \in \Sigma^*\), \(w_1s \in L\) implies \(w_2s \in L\).  If
\(|w_1| \leq k_1 + k_2\), then \(w_1 = w_2\), and this is immediate.  Let then
\((M, E)\) be the finite object associated with both \(w_1\) and \(w_2\) and assume
\(w_1s \in L\).  If \(|s| < k_3\), then \(s \in M\), and thus \(w_2s \in L\).  Otherwise,
consider a factorization \(w_1 = u_1x_1v_1\) and \(s = y_1z\) with \(|z| = k_3\) such
that there is a run with label \(w_1s\) in which \(x\) and \(y\) are not coupled ---
again, this exists by hypothesis.  Let \(c\) be the min-endcap of this run over
\(u_1x_1v_1yz\), and \(e\) be its stack-effect.  Since \((c, [e]_{k_3}) \in E\), there is another suffix \(s'\) such that
\(w_2 = u_2x_2v_2\) and \(u_2x_2v_2s'\) is accepted by a run with min-endcap
\(c\) and stack-effect \(e'\) with \(e\) and \(e'\) being \(k_3\)_suffix-equivalent.
\Cref{lem:surgery} allows us to conclude, noting that we enforced that \(e =
\eps\) iff \(e' = \eps\) by letting \(\eps\) be in its own equivalence class in
\Cref{lem:finclass}.

\section{Proof of the Deterministic Coupling Lemma}\label{sec:coupdpda}

\def\jas{Jan\v{c}ar and \v{S}\'ima\xspace}

\subsection{Pumping and coupling}

We first provide pumping lemmata pertaining to couplings, starting with a
variant of the classic pumping lemma for CFLs: If a word can be factored into
two unbalanced parts, such that one side is larger than the other, there must
also be a pumpable subsection which is similarly unbalanced.

\begin{restatable}{lemma}{LemUnbalancedOgden}\labelandarrows{lem:unbalanced_ogden}
  let \(\lambda\) and \(\mu\) be two words over an alphabet \(\Sigma\) not containing \(\#\), \(\lambda'\) a suffix of \(\lambda\) and \(\mu'\) a prefix of \(\mu\).
  Let  \(L\subseteq \lambda'\lambda^{*}\#\mu^{*}\mu'\) be a CFL.
  There is a constant \(p\) such that every \(\lambda'\lambda^{n}\#\mu^{n'}\mu'\in L\) with \(n\geq p\) and \(n>n'+1\) admits integers \(0\leq l<k\leq p\) such that for all \(m\geq 0\), \(\lambda'\lambda^{n+km}\#\mu^{n'+lm}\mu'\) is in \(L\).
\end{restatable}

With this tool in hand, we show how to use it on a word with some forbidden
couplings for a given run.

\begin{restatable}{lemma}{LemPumpingTwo}\labelandarrows{lem:pumping2}
  Let \(A\) be a PDA for a language \(L\). For every \(\lambda,\mu\), there is a constant \(p\) such that if
  \begin{itemize}
    \item \(\lambda'\) is a suffix of \(\lambda\) and \(\mu'\) is a prefix of \(\mu\), and
    \item \(r\) is a run labeled by \(w=u\lambda'\lambda^{n}v\mu^{n'}\mu'z\in L\) for any \(u,v,z\) and \(n-n'\geq p\), and
    \item \(\lambda'\lambda^{n}\) and \(\mu^{n'}\mu'\) are not \(r\)-coupled with any of \(u\), \(v\) and \(z\),
  \end{itemize}
then there are \(0\leq l<k\leq p\) such that for all \(m\geq 0\), \(u \lambda' \lambda^{n+km}v\mu^{n'+lm} \mu'z\) is in \(L\).
\end{restatable}

We can modify~\cref{lem:pumping2} so that the conclusion stands when the precise
coupling required by the decoupling lemma fails to hold.  For a fixed PDA \(A\),
and two words \(\lambda\) and \(\mu\), we define the pumping constant
\(\pc(\lambda,\mu)\) to be the maximum of the constants \(p\) given by~\cref{lem:pumping2}
in the cases \((\lambda,\mu)\) and \((\lambda,\varepsilon)\).

\begin{restatable}{lemma}{LemFindingSpots}\labelandarrows{lem:finding_spots}
  Let \(A\) be a PDA with a total of \(|A|\) states, and let
  \(\alpha,\lambda,\beta,\mu,\gamma\in \Sigma^{+}\).  Let \(p = 5\cdot
  \pc(\lambda,\mu)\cdot |\alpha\beta\gamma|\cdot|A|\), and \(u=\alpha\),
  \(x=\lambda^{p}\), \(v=\beta\), \(y=\mu^{p}\) and \(z=\gamma\), assume that
  \begin{itemize}
    \item there is an accepting run \(r\) labeled by \(w=ux^{n}vy^{n}z\in L\), and
          \item there is \(1\leq i\leq n\) for which the \(i\)-th occurrence of \(x\) is \emph{not} \(r\)-coupled with the \((n-i+1)\)-th occurrence of \(y\).
  \end{itemize}

  Then there are \(0\leq l<k\leq \pc(\lambda,\mu)\) such that for all \(m\geq 0\), \(ux^{n}\lambda^{km}vy^{n}\mu^{lm}z\) is in \(L\).
  There are also \(0\leq k<l\leq \pc(\lambda,\mu)\) such that for all \(m\geq 0\), \(ux^{n}\lambda^{km}vy^{n}\mu^{lm}z\) is in \(L\).
\end{restatable}

\subsection{Proof of the Deterministic Coupling Lemma}

We start by stating a deep result of \jas,\footnote{The result also appears in
  \cite{BordihnM20}, where they credit~\cite{Stearns67}, although a proof seems
  hard to extract from there.} which, at its core, says that \dcflp languages
all embed \(\{a^nb^n \mid n \geq 0\}\) in some way:
\begin{lemma}[{\cite[Theorem 1]{jancars21}}]\label{lem:js}
  Let \(L\) be a language in \(\dcflp\).
  There exist \(  \alpha,\lambda,\beta,\mu,\gamma \in \Sigma^{*}\) and a set \(S\) among
  \begin{itemize}
    \item \(S_{=}=\{ n,m \ |\ n = m \}\),
    \item \(S_{\geq}=\{ n,m \ |\ n\geq m\}\),
    \item \(S_<=\{ n,m \ |\ n< m\}\),
    \item \(S_{\neq}=\{ n,m \ |\ n \neq m\}\),
  \end{itemize}
  such that \(\alpha\lambda^{n}\beta\mu^{m}\gamma\in L\) is equivalent to \((n,m)\in S\).
\end{lemma}

Now recall the assumptions of the Deterministic coupling lemma: \(L\) is a DCFL recognized by a (not necessarily deterministic) PDA \(A\).
It is such that all sublanguages \(\{ux^nvy^nz \mid n \geq 0\} \subseteq L\) there is an \(n\) and \(1 \leq i \leq n\) such that, in \(ux^nvy^nz\), the \(i\)-th occurrence of \(x\) is \emph{not} \(A\)_coupled with the \((n-i+1)\)-th occurrence of \(y\).
  By contradiction, assume that \(L\) is not regular.
  We start by applying~\cref{lem:js}: there are \(  \alpha,\lambda,\beta,\mu,\gamma' \in \Sigma^{*}\) and a set \(S\) among the four possibilities such that \(\alpha\lambda^{n}\beta\mu^{m}\gamma'\in L\) if and only if \((n,m)\in S\).
  Let \(\gamma =  \mu^{\pc(\lambda,\mu)!}\gamma'\) if \(S= S_{\neq}\), \(\gamma=\mu\gamma'\) if \(S=S_{<}\) and \(\gamma=\gamma'\) otherwise.
  Let \(p=5\cdot \pc(\lambda,\mu)\cdot |\alpha\beta\gamma|\cdot |A|\), and let \(u=\alpha\), \(x=\lambda^{p}\), \(v=\beta \), \(y=\mu^{p}\) and \(z=\gamma\).
  Notice that for any choice of \(S\) and \(n\), the word \(ux^{n}vy^{n}w\) is in \(L\) for any \(n\).
  By assumption on \(L\), there is an \(n\) and \(1\leq i\le n\) such that, in \(ux^{n}vy^{n}z\), the \(i\)-th occurrence of \(x\) is not \(A\)-coupled with the \((n-i+1)\)-th occurrence of \(y\).
  Thus, there is an accepting run \(r\) for which this coupling does not stand.
  Thus all conditions of~\cref{lem:finding_spots} stand: there are \(l,k\) such
  that for all \(m\geq 0\), \(ux^{n}\lambda^{km}vy^{n}\mu^{lm}z\) is in \(L\).
  For \(S_{\geq}\), we take \(0\leq k<l\leq \pc(\lambda,\mu)\); and  \(0\leq
  l<k\leq \pc(\lambda,\mu) \) for all other cases.  We proceed by case over all
  possibilities for \(S\).
  \begin{itemize}
    \item \textbf{Case \(S_{=}\).} With \(m=1\), \(ux^{n}\lambda^{k}vy^{n}\mu^{l}z= \alpha\lambda^{np+k}\beta\mu^{np+l}\gamma'\) is in \(L\).
          Because of~\cref{lem:js}, this implies \(np+k=np+l\) and thus \(k=l\); a contradiction.
    \item \textbf{Case \(S_{\geq}\).} With \(m=1\), \(ux^{n}\lambda^{k}vy^{n}\mu^{l}z= \alpha\lambda^{np+k}\beta\mu^{np+l}\gamma'\) is in \(L\).
          Because of~\cref{lem:js}, this implies \(np+k\geq np+l\) and thus \(l\leq k\); a contradiction.
    \item \textbf{Case \(S_{<}\).} With \(m=1\), \(ux^{n}\lambda^{k}vy^{n}\mu^{l}z= \alpha\lambda^{np+k}\beta\mu^{np+l+1}\gamma'\) is in \(L\).
          Because of~\cref{lem:js}, this implies \(np+k< np+l+1\) and thus \(k\leq l\); a contradiction.
    \item \textbf{Case \(S_{\neq}\).} With \(m=\frac{\pc(\lambda,\mu)!}{k-l}\)
    (which is an integer because \(0<k-l\leq \pc(\lambda,\mu)\)),
      the word \(ux^{n}\lambda^{km}vy^{n}\mu^{lm}z= \alpha\lambda^{np+km}\beta\mu^{np+lm+\pc(\lambda,\mu)!}\gamma'\) is in \(L\).
      Because of~\cref{lem:js}, this implies \(np+km\neq np+lm+\pc(\lambda,\mu)!\) and thus \((k-l)m\neq \pc(\lambda,\mu)!\); a contradiction.
  \end{itemize}
  We reach a contradiction at the end of every road. Thus, \(L\) must be regular.

\begin{remark}
  It may be tempting to follow a similar approach for CFL, by extending \jas's
  result to \cflp.  However, a corollary of \Cref{lem:js} is that every
  nonregular DCFL embeds, in some sense, \(L_\#=\{a^nb^n \mid n \geq 0\}\).  In
  \cite{jancars21}, this is formalized using Mealy machines with
  oracles:\footnote{The exact definition of Mealy machine with oracle is not
    important to our discussion.}  for every \(L \in \dcflp\), there is such a
  machine \(A\) such that \(L_\#\) is the language of \(A\) with oracle
  \(L\).  We show that it is not possible to obtain a similar result for
  nonregular CFL.  Call a Turing machine with oracle \emph{regularity
    preserving} if its language is regular when its oracle is.  Consider \(\cC\)
  to be a recursively enumerable set of regularity preserving Turing machines
  with oracles that always halt (for instance, the set of Mealy machines).
  Then:
\end{remark}
\begin{proposition}\label{prop:nosimplecflp}
  There is no single nonregular CFL \(S\) such that for any other nonregular CFL
  \(L\), there is \(M \in \cC\) such that \(S\) is the language of \(M\) with oracle \(L\).
\end{proposition}

\begin{proof}
  The proof hinges on the fact that regularity for CFL is a
  \(\Sigma_2\)_complete problem in the arithmetic hierarchy~\cite[p.429]{BordihnHK05}.
  Assume for a contradiction that there is such an \(S\).  Then, given a CFL
  \(L\), one can enumerate all the machines \(M \in \cC\), and check, universally,
  that \(S\) is the language of \(M\) with oracle \(L\) (for all words, check that
  they are both accepted or rejected, the latter property being decidable by
  hypothesis).  This will hold if, and only if, \(L\) is nonregular, showing that
  regularity for CFL would be in \(\Pi_2\), and the arithmetic hierarchy would
  collapse, a contradiction.
\end{proof}

\section{Shuffles and Their Interactions with Couplings}\label{sec:shuf}

\subsection{Definitions: trajectories and shuffles}

Let \(L_1\) and \(L_2\) be two languages that we will always assume to be over
disjoint alphabets in the context of shuffles.  Their \emph{full shuffle}
\(L_1 \shuf L_2\) is defined as the set of words \(u_1v_1\cdots u_nv_n\) such that
\(u_1\cdots u_n \in L_1\) and \(v_1\cdots v_n \in L_2\), where each of the
\(u_i\), \(v_i\) can be empty.

A \emph{trajectory} is a regular language \(T\) over the alphabet \(\{s, t\}\).
Let \(g\) be the morphism that maps the alphabet of \(L_1\) to \(s\) and that of
\(L_2\) to \(t\).  The \emph{shuffle according to the trajectory \(T\)} of
\(L_1\) and \(L_2\) is defined as \(L_1 \shuf_T L_2 = (L_1 \shuf L_2) \cap
g^{-1}(T)\).  If \(c \in \{s, t\}^*\), we write \(\shuf_c\) for
\(\shuf_{\{c\}}\) and call \(c\) itself a trajectory.

For \(\cC\) a class of languages (we will use \(\cC = \cflp\) or
\(\cC = \dcflp\)), a trajectory \(T\) is:
\begin{itemize}[noitemsep,topsep=0pt,parsep=0pt,partopsep=0pt]
\item \(\cC\)_safe if for all \(L_1, L_2 \in \cC\), \(L_1 \shuf_T L_2\) is CFL;
\item \(\cC\)_hostile if for all \(L_1, L_2 \in \cC\), \(L_1 \shuf_T L_2\) is \emph{not}
  CFL;
\item \(\cC\)_mixed otherwise, i.e., there are languages of \(\cC\) that shuffle
  into a CFL and others into a nonCFL.
\end{itemize}

Our ultimate goal is the classification of trajectories according to their
safety and hostility to \cflp and \dcflp.

\subsection{Resilient couplings and how to find them}

In order to use the Coupling Lemma within the context of shuffles, we
will explore how couplings of a language survive within a shuffle.  We start by
defining what ``survival'' means and prove, in the forthcoming Resilience
Lemmas, that any PDA with a nonregular language has couplings that survive
shuffles.

Let \(w_1 = xyz\) and \(w_2\) be two words on disjoint alphabets, \(c\) a trajectory,
and let \(w = w_1 \shuf_c w_2\).  We write \(w[y]\) for the smallest substring of
\(w\) that contains all of \(y\).  More formally, write \(w = x'y'z'\) such that
\(x'\) contains \(x\) (i.e., \(x\) is a scattered subword of \(x'\)), \(y'\) contains
\(y\), and \(z'\) contains \(z\), in such a way that \(y'\) starts and ends with the
same letters as \(y\).  Then \(w[y] = y'\).  We similarly use this notation for
parts of \(w_2\), which is not ambiguous since they are on disjoint alphabets.

\begin{plaindef}
  Let \(A\) be a PDA recognizing \(L = L_1 \shuf_T L_2\) for a trajectory \(T\)
  and let \(w_1 = uxvyz \in L_1\).  We say that \(x\) is \emph{resiliently}
  \(A\)_coupled with \(y\) if for any \(c \in T\) and any \(w \in w_1 \shuf_c
  L_2\), \(w[x]\) is \(A\)_coupled with \(w[y]\).  We use the same vocabulary
  for similar words in \(L_2\).
\end{plaindef}

The Coupling and Deterministic Coupling Lemma give rise to two separate lemmas
extracting pairs of resiliently coupled words, which exist when the trajectory
doesn't brutely filter either of the two languages into a different class via
word length. Let \(\pi_1: (\Sigma_1 \cup \Sigma_2)^* \to \Sigma_1^*\) be the
natural projection with \(\pi_1(h) = h\) for all \(h \in \Sigma_1\), and
\(\eps\) otherwise. Defining \(\pi_2\) similarly, we have the following:

\begin{lemma}[Resilience Lemma]
  Let \(A\) be a PDA recognizing \(L = L_1 \shuf_T L_2\) for a trajectory \(T\)
  with \(L_1\) in \(\cflp\). If \(\pi_1(L_1 \shuf_T L_2)\) is nonregular, then,
  for all \(k_1, k_2, k_3 \in \bbN\), there is a word \(uxvyz \in L_1\) with
  \(|u| = k_1, |v| = k_2, |z| = k_3\), and \(x\) is resiliently \(A\)_coupled
  with \(y\). The same holds for \(L_2\) with respect to \(\pi_2\), defined
  analogously.
\end{lemma}
\begin{proof}
  We demonstrate that \(\pi_1(L_1 \shuf_T L_2)\) is always a CFL. Then, with
  the nonregularity condition, we may apply the coupling lemma. Let \(\pkh\colon
  \{s,t\}^* \to \bbN^2\) be the Parikh morphism defined by \(\pkh(s) = (1, 0)\)
  and \(\pkh(t) = (0,1)\), extended to languages where convenient. Parikh's
  theorem~\cite{Parikh66} notes that for every context free language there is a
  regular one with equal image under this morphism, and so \(\pkh(R_2) =
  \pkh(L_2)\) for some regular \(R_2\). Further, let \(R_1 = \{w_1 \in
  \Sigma_1^* \mid \exists w_2 \in R_2, |w_1| = |w_2|\}\), which is similarly a
  regular language. Then, we have it that \(\pi_1(L_1 \shuf_T R_2)\) =
  \(\pi_1(L_1 \cap R_1) = L_1 \cap R_1\) by construction. As \cfl is closed
  under intersection with a regular language, we have it that \(\pi_1(L)\)
  remains in \(\cfl\).

  Let \(A_p\) be the PDA obtained by applying \(\pi_1\) to each transition
  label, mapping every letter of \(\Sigma_2\) to \(\eps\). Naturally, \(A_p\)
  accepts \(L_p = \pi_1(L_1 \shuf_T L_2)\) which by premise is a nonregular
  sublanguage of \(L_1\).  Since it is nonregular, the Coupling Lemma tells us
  that for every \(k_1, k_2, k_3 \in \bbN\), there is a word \(w_1 = uxvyz \in
  L_p\) with \(|u| = k_1, |v| = k_2, |z| = k_3\), for which \(x\) and \(y\) are
  \(A_p\)_coupled. Given any word \(w_2 \in \Sigma_2^*\), any accepting run for
  \(w = w_1 \shuf_T w_2\) in \(A\) can be projected into an accepting run in
  \(A_p\) for \(w_1\) in which \(x\) and \(y\) are coupled. Projecting back,
  \(w[x]\) and \(w[y]\) are coupled in \(r\), and as this holds for all runs,
  \(w[x]\) and \(w[y]\) are \(A\)_coupled, concluding the proof.
\end{proof}

Using the same technique, when \(L_1 \in \dcflp\), we can find a uniform family of
resiliently coupled words using the Deterministic Coupling Lemma:
\begin{lemma}[Deterministic Resilience Lemma]\label{lem:detcoupres}
  Let \(A\) be a PDA recognizing \(L = L_1 \shuf_T L_2\) for a trajectory \(T\)
  with \(L_1 \in \dcflp\). If \(\pi_1(L_1 \shuf_T L_2)\) is nonregular, then,
  there is a sublanguage \(\{ux^nvy^nz \mid n \geq 0\} \subseteq L_1\) such that
  for every \(n\) and \(i \leq n\), in the word \(ux^nvy^nz\), the \(i\)-th
  occurrence of \(x\) is resiliently \(A\)_coupled with the \((n-i+1)\)-th
  occurrence of \(y\). The same holds for \(L_2\), analogously.
\end{lemma}

\subsection{Resilient couplings and CFL\('\)_hostility}

\begin{theorem}\label{thm:cflp-hostile}
  The trajectories \(s^*t^*s^*t^*\), \((st)^*t^*(st)^*t^*\) are \cflp_hostile.
  Any trajectory that includes them (e.g., \((st^+)^*\)) is also \cflp_hostile.
\end{theorem}
\begin{proof}
  Note that these particular trajectories have the property that \(\pi_1(L) =
  L_1\) and \(\pi_2(L) = L_2\), and so either projection is nonregular for all
  nonregular languages. Hence, shuffle resilience is always available.

  \emph{(\(T = s^*t^*s^*t^*\)).}\quad Let \(L_1, L_2\) be two nonregular CFLs and write
  \(L = L_1 \shuf_T L_2\).  Assume, towards a contradiction, that \(L\) is
  recognized by a PDA \(A\).  We apply the Resilience Lemma with
  \(k_1 = k_2 = k_3 = 0\) to both \(L_1\) and \(L_2\).  This provides, for
  \(i = 1,2\), a word \(w_i = x_iy_i \in L_i\) such that \(x_i\) and \(y_i\) are
  resiliently \(A\)_coupled.  The word \(x_1x_2y_1y_2\), which is in \(L\), thus has
  \(x_1\) \(A\)_coupled with \(y_1\) and \(x_2\) \(A\)_coupled with \(y_2\); since couplings
  cannot cross, this is a contradiction.

  \emph{(\(T = (st)^*t^*(st)^*t^*\)).}\quad Let \(L_1, L_2\) be two nonregular
  CFLs and write \(L = L_1 \shuf_T L_2\).  Assume, towards a contradiction, that
  \(L\) is recognized by a PDA \(A\).  We apply the Resilience Lemma with \(k_1
  = k_2 = k_3 = 0\) to \(L_1\), providing a word \(w_1 = x_1y_1 \in L_1\) with
  \(x_1\) and \(y_1\) resiliently \(A\)_coupled.  Set now \(k_1 = |x_1|, k_2 =
  |y_1|, k_3 = 0\).  By the Resilience Lemma for \(L_2\), there is \(w_2 =
  u_2x_2v_2y_2 \in L_2\) with \(x_2\) and \(y_2\) resiliently \(A\)_coupled.
  Let \(S = (st)^*\).  Consider the word: \[w = (x_1 \shuf_S u_2)x_2(y_1 \shuf_S
  v_2)y_2.\] Since \(|x_1| = |u_2|\) and \(|y_1| = |v_2|\), \(w\) is well
  defined and \(w \in L\).  But the resilient couplings cross, a contradiction.
\end{proof}

\subsection{Resilient couplings and \(\text{DCFL}'\)_hostility}

The goal of this section is to prove a characterization of
\(\text{DCFL}'\)_hostility, based on the automaton for the trajectory.  We also
show that this characterization is decidable and briefly outline related
complexities. We start by identifying patterns within automata, and then show
that they can be used to characterize \(\text{DCFL}'\)_hostility.

\subsubsection{Hard and easy strongly-connected components in DFA}

Let \(A\) be a DFA over \(\{s, t\}\).  A (recurrent) strongly-connected
component (SCC) in \(A\) is a set of states \(P\) such that \(\forall p, q \in
P\), there is a nonempty path from \(p\) to \(q\).  An \emph{SCC-pattern} in
\(A\) is a sequence \(\cS = u_1P_1u_2\ldots u_nP_nu_{n+1}\) such that \(P_1
\ldots P_n\) are SCC in \(A\), each \(u_i\) is a path in \(A\), \(u_1u_2\ldots
u_{n+1}\) is also a path in \(A\), and \(u_i\) goes from a state in \(P_{i-1}\)
to a state in \(P_i\).  The \emph{paths of} \cS are the paths of \(A\) that
start with \(u_1\), then a path within \(P_1\), then \(u_2\), and so forth.  The
SCC-pattern is \emph{accepting} if \(u_1\) starts in the initial state of \(A\)
and \(u_{n+1}\) ends in an accepting state.  Classically:
\begin{proposition}\label{prop:scc}
  There is a finite computable set of accepting SCC-patterns such that the set
  of accepting paths of \(A\) is equal to the union of the paths of the
  SCC-patterns.
\end{proposition}

The \emph{alphabet} of an SCC is the set of labels of transitions internal to
the SCC.  We will talk of \(s\)-SCC, \(t\)-SCC, and \(st\)-SCC for SCCs with alphabets
\(\{s\}\), \(\{t\}\) and \(\{s,t\}\), respectively.  An SCC-pattern is \emph{hard} if
it contains either an \(st\)-SCC or two \(s\)-SCCs and two \(t\)-SCCs that are
interleaved in the pattern (as \(stst\) or \(tsts\)).

We fix an arbitrary SCC-pattern decomposition of the accepting paths of \(A\), as
provided by \Cref{prop:scc}.  We call \(\hard(A)\) the set of labels of the paths of the
hard SCC-patterns of \(A\), and \(\easy(A)\) for the rest.

We further make use of common techniques in the study of semilinear sets.
However, we only need such sets over 2 dimensions, ultimately corresponding to the
two letters $s$ and $t$ above. This case lends itself well to geometric
intuition on the plane, and hence, we use the following definition from
\cite{nguyen2018polyhedra}: A set is \emph{semilinear} if it can be defined as
\(\{(x, y) \mid \bbN \models \phi(x, y)\}\) with \(\phi\) a DNF where each
conjunct is \((cx + dy + e > 0)\) or \((cx + dy + e \equiv 0 \text{ mod } p)\)
for constants \(c, d, e, p\).
Any formula of the form \(cx + dy + e > 0\) corresponds to a half-plane delimited
by the line defined by \( cx + dy + e = 0\).
Thus, a clause corresponds to a polyhedron from which periodic points are selected.
We will have a closer look at the unbounded polyhedra associated to a formula and classify them in \emph{cones} and \emph{bands}.
The precise definitions are only needed for the proof of~\cref{prop:decide_antigrids}, and are therefore deferred here.

In particular, we will be interested in the following property of the semilinear
set \(\pkh(\hard(A))\), recalling that \(\pkh\colon \{s,t\}^* \to \bbN^2\) is
the Parikh morphism defined by \(\pkh(s) = (1, 0)\) and \(\pkh(t) = (0,1)\). An
example of this property is presented in \Cref{fig:grids}.
\begin{plaindef}
  We call \emph{periodic} any set \(i+p\bbN\) for \(i,p\in\bbN\).
  Let \(S \subseteq \bbN^2\) be a semilinear set.  A \emph{grid} for \(S\) is a pair \(U, V \subseteq
  \bbN\) of periodic sets such that \(U \times V \subseteq S\).  An \emph{antigrid} for \(S\) is a
  grid for the complement of \(S\).
\end{plaindef}

\begin{figure}[!h]
  \centering
  \begin{tikzpicture}[>=stealth,thick]

\begin{scope}[scale=0.8, xshift=-5cm]
    \fill[fill=gray!42] (0, 0) -- (0, 2) -- (4, 6) -- (6, 6) -- (6, 4) -- (2, 0);

    \draw[->, solid, thick] (0,0) -- (0,6.5);
    \draw[->, solid, thick] (0,0) -- (6.5,0);
    \foreach \y in {1, 2, 3, 4, 5} {
        \draw[solid, thin] (-0.1, \y) -- (0.1, \y);
        \draw[dotted, very thin] (-0.1, \y) -- (6, \y);
    }
    \foreach \x in {1, 2, 3, 4, 5} {
        \draw[solid, thin] (\x, -0.1) -- (\x, 0.1);
        \draw[dotted, very thin] (\x, -0.1) -- (\x, 6);
    }

    \draw[->, dashed, thick] (0, 2) -- (4, 6);
    \draw[->, dashed, thick] (2, 0) -- (6, 4);
    \foreach \x in {1, 3, 5} {
        \draw[->, solid, thin] (\x, 0) -- (\x, 6);
    }

    \filldraw[fill=white] (1, 0) circle (4pt);
    \filldraw[fill=white] (1, 1) circle (4pt);
    \filldraw[fill=white] (1, 2) circle (4pt);
    \filldraw[fill=white] (3, 2) circle (4pt);
    \filldraw[fill=white] (3, 3) circle (4pt);
    \filldraw[fill=white] (3, 4) circle (4pt);
    \filldraw[fill=white] (5, 4) circle (4pt);
    \filldraw[fill=white] (5, 5) circle (4pt);

    \node[left] at (-0.25,6) {$y$};
    \node[above] at (6,0) {$x$};
    \node[fill=white] at (1.5, 4.5) {Points in $S$};
\end{scope}

\begin{scope}[scale=0.8, xshift=3cm]
    \fill[fill=gray!42] (0, 2) -- (4, 6) -- (0, 6);
    \fill[fill=gray!42] (2, 0) -- (6, 4) -- (6, 0);

    \draw[->, solid, thick] (0,0) -- (0,6.5);
    \draw[->, solid, thick] (0,0) -- (6.5,0);
    \foreach \y in {1, 2, 3, 4, 5} {
        \draw[solid, thin] (-0.1, \y) -- (0.1, \y);
        \draw[dotted, very thin] (-0.1, \y) -- (6, \y);
    }
    \foreach \x in {1, 2, 3, 4, 5} {
        \draw[solid, thin] (\x, -0.1) -- (\x, 0.1);
        \draw[dotted, very thin] (\x, -0.1) -- (\x, 6);
    }

    \draw[->, solid, thick] (0, 2) -- (4, 6);
    \draw[->, solid, thick] (2, 0) -- (6, 4);
    \foreach \x in {0, 2, 4, 6} {
        \draw[->, solid, thin] (\x, 0) -- (\x, 6);
    }

    \foreach \x in {0, 2, 4, 6} {
        \foreach \y in {0, 2, 4, 6} {
            \filldraw[fill=gray] (\x, \y) circle (4pt);
        }
    }
    \node[left] at (-0.25,6) {$y$};
    \node[above] at (6.5, 0) {$x$};
    \node[fill=white] at (1.5, 4.5) {Antigrid for $S$};
\end{scope}

\end{tikzpicture}
  \caption{
    A semilinear set $S$, along with an antigrid of $S$, where \(S := (x - y + 2
    > 0)\) \(\land\ (-x + y + 2 > 0)\) \(\land\ (x + 1 \equiv 0 \text{ mod } 2)\).
    Notice that in the graph of $S$, the inequalities form a "band", while in
    the complement, this becomes a pair of "cones". In either case, the modulus
    condition filters these into "windows" of alternating parity. It is further
    worth noting that $S$ cannot have a grid itself; only antigrids.
  }
  \label{fig:grids}
\end{figure}
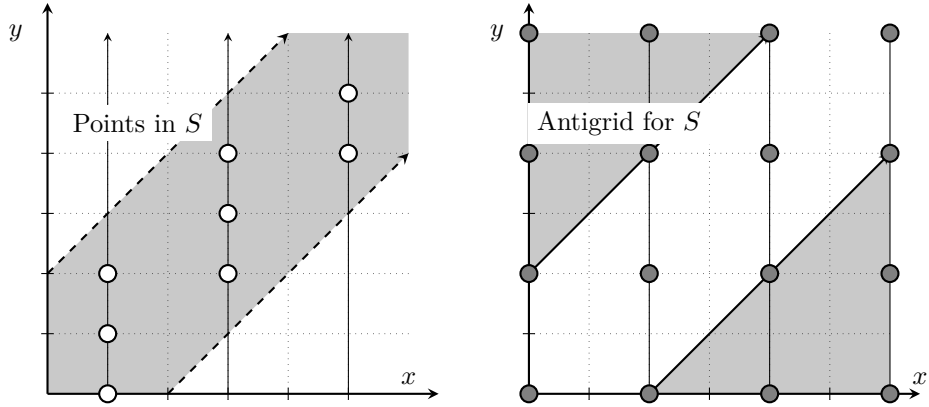

Our forthcoming characterization of \(\text{DCFL}'\)_hostility will focus on
finding antigrids in \(\pkh(\hard(A))\); we note that this is a decidable problem:
\begin{restatable}{proposition}{PropDecideAntigrids}\label{prop:decide_antigrids}
  It is decidable whether a given semilinear set contains a grid or an antigrid.
\end{restatable}
\begin{proof}[Proof sketch.]
  Since semilinear sets are effectively closed under complement, it is sufficient to
  argue that we can find grids in them.  Consider a representation of a given
  \(S \subseteq \bbN^2\) in the plane.  When sufficiently away from the origin
  to avoid any bounded polyhedra, \(S\) is the union of affine cones and bands
  (or, more precisely, \(S\) is the union of periodic windows within these cones
  and bands). The gaps left by the cones of \(S\), seen continuously, are also
  cones and bands.  Let us call the former \(S\)-cones and \(S\)-bands, and the
  latter non-\(S\)-cones and non-\(S\)-bands.

  If there is a non-\(S\)-cone, then no grid can exist.  Assume otherwise.  Notice
  that if there is a grid \(U \times V\), then any infinite periodic set
  \(U' \subseteq U\) is such that \(U' \times V\) is also a grid.  We can thus
  assume that the period of the grid is such that it ``jumps'' over the bands,
  and thus, there are no bands.  Now, a grid exists if and only if the
  \(S\)-cones are compatible, in the sense that extending the periodic windows
  of the \(S\)-cones to the whole plane, they have a nonempty intersection.  All
  of these properties are decidable.
\end{proof}

\begin{remark}\label{rem:complexity}
The previous decidability algorithm yields a \(\PSPACE\) upper bound. Briefly,
this follows from the fact that while the product of periods of \(S\) may be
exponential, and so too the regions to look for "windows" to line up, this
problem admits a \(\coPSPACE\) solution by guessing a point in the intersection.
Conclude by Savitch's theorem. In the special case where the periods are all
related (or equal, in the best case), this algorithm is only \(PTIME\).
\end{remark}

\subsubsection{Characterization of \(\text{DCFL}'\)_hostility}

\begin{theorem}\label{thm:dcflhost}
  Let \(T\) be a regular trajectory recognized by a DFA \(A\).
  \begin{enumerate}
  \item \(\pkh(\hard(A))\) does not have an antigrid if, and only if, \(\shuf_T\)
    is \dcflp_hostile;
  \item \(\pkh(\hard(A))\) is nonempty and has an antigrid if, and only if,
    \(\shuf_T\) is \dcflp_mixed;
  \item \(\pkh(\hard(A))\) is empty if, and only if,
    \(\shuf_T\) is \cflp_safe.
  \end{enumerate}
\end{theorem}

All the conditions in the theorem are decidable.  Indeed, \(\hard(A)\) is again
regular and therefore \(\pkh(\hard(A))\) is (effectively) semilinear by Parikh's
theorem~\cite{Parikh66}.  Thus we can invoke~\cref{prop:decide_antigrids}.

The theorem is proved through a series of lemmas.

\begin{lemma}\label{lem:safe}
  If \(\pkh(\hard(A)) = \emptyset\), then \(\shuf_T\) is \cflp_safe.
\end{lemma}
\begin{proof}
  The condition is equivalent to \(\hard(A)\) being empty.  Thus, \(A\) can be
  decomposed into a disjoint union of SCC-patterns of the form:
  \[\underbrace{u_0U_1u_1\cdots U_ku_k}_{u \text{ block}} \cdot
  \underbrace{v_0V_1v_1\cdots V_\ell v_\ell}_{v \text{ block}} \cdot \underbrace{w_0W_1w_1\cdots W_m
    w_m}_{w \text{ block}}\]%
  where the \(U_i\) and \(W_i\) are all \(s\)-SCC and the \(V_i\) all \(t\)-SCC
  (or vice versa).  Given two CFL \(L_1, L_2\) over \(\Sigma_1, \Sigma_2\)
  recognized by PDAs \(A_1\) and \(A_2\), the shuffle \(L_1 \shuf_T L_2\) can be
  simulated using the standard technique used to simulate substring insertion
  (e.g., the trajectory \(s^*t^*s^*\)), with the addition of storing the
  (finite) letters of \(\Sigma_2\), \(\Sigma_1\) which occur in the \text{u},
  \text{v} blocks respectively, while using $k$-suffix equivalence to handle the
  letters of \(\Sigma_2\) which occur in the \text{w} block.

  For each pattern, let \(x, y, z\) be defined as the letters of \(\Sigma_2\) in
  the \(u\) block, \(\Sigma_1\) in the \(v\) block, and \(\Sigma_2\) in the
  \(w\) block, respectively. These are all finite, as their letters can only
  occur outside of the SCC within their block.  We construct a PDA in 5 parts as
  follows: Let \(A'_1\) be the PDA over states \(Q_{A_1} \times \Sigma^{|x|}_2\)
  which simulates \(A_1\), only storing up to \(|x|\) letters of \(\Sigma_2\) in
  a state; define \(A'_3\) over \(\{Q_{A_2} \times \Sigma^{|y|}_1\}\) similarly.
  Further, let \(A'_2\) be the PDA over \(Q_{A_2} \times \Sigma^{|x|}\) which
  simulates \(A_2\), which, rather than reading letters from the word, "reads"
  the letters in \(\Sigma_2\) from the states along \(\eps\) transitions.
  Define \(A_4\) over \(\{Q_{A_1} \times \Sigma^{|y|}_1\}\) similarly.

  Recall that two stack values are called \(k\)-suffix equivalent if they have
  identical behavior over \(k\) letters, that this forms an equivalence relation
  with a finite number of classes per \cref{lem:finclass}, and let \(S\) be
  labels for the set of \(|z|\)-suffix equivalent stacks on \(A_2\). Then, let
  \(A_6\) be the PDA over \(\{Q_{A_1} \times \Sigma^{|z|}_1\} \cup S\) which
  simulates \(A_1\) while storing accumulated letters from \(z\) in states,
  checking their status against a (preset) \(S\) at run completion. We similarly
  need a component to translate an arbitrary stack into the relevant
  suffix-equivalent classes. Let \(A_5\) be a new PDA over \(\{Q_2 \times
  \Sigma^{|z|}\} \cup S\) with only \(\eps\) transitions which, given an
  initial stack value, simulates each possible start-state times \(|z|\)
  combination to get the possible end states, the whole of which then maps to
  \(S\) for (finite) storage.

  To assemble these into the final PDA, take the concatenation \(A' = A'_1 A'_2
  \ldots A'_6\), where \(A'\) uses \(A'_1\) to read the \(u\) block, inserts a
  new bottom of stack symbol \$ in order to preserve the stack from \(L_1\), and
  then \(A'_2 A'_3 A'_4\) to read the \(v\) block. In order to retrieve the
  beginning of \(A_1\)'s stack, use \(A'_5\) to remove any values on the stack
  until arriving back at \$, compressing the result into \(S\); lastly, remove
  the \$ and run \(A'_6\).
\end{proof}

\noindent We import the well known Ogden's lemma, stated in
automata-theoretic terms:
\begin{lemma}\label{lem:ogden}
Suppose \(L\) is a CFL recognized by \(A\). Then, there is some \(p\) such that,
for any word \(w\in L\), if \(p\) any letters are 'marked' then,
there is a factorization \(w = uxvyz\) such that \(x, y\) contain between 1 and
\(p\) marked positions, and for any \(n\), the run \(r \langle u
\rangle r \langle x \rangle^n r \langle v \rangle r \langle y \rangle^n r
\langle z \rangle\) is both valid and accepting.
\end{lemma}

\begin{lemma}\label{lem:anbna} Consider \(L = \{a^n b^{*}a^n b^* \mid n \in
\mathbb{N}\}\) and some \(L'\).  For any PDA \(A\) recognizing \(L\shuf L'\), there is a constant \(k\)
such that, for any \(w=a^{n}b^{k_{1}}a^{n}b^{k_{2}} \in L\) with \(n,k_{1},k_{2}\geq 1\), the first \(a^{n}\) has a resilient \(A\)-coupling with the second \(a^{n}\).
\end{lemma}
\begin{proof}
  We prove the existence of \(A\)-couplings in any PDA \(A\) for \(L\), the resilience
  falling from the same techniques already used.
  We employ elementary methods given the simplicity of these languages.  Mark
  the first \(p\) letters of \(w\), i.e, only \(a\), and apply
  \Cref{lem:ogden}. Then, any run \(r\) of \(w\) must have a factorization \(r =
  r_1r_2r_3r_4r_5\) with \(r_2, r_4\) containing at least one letter of the first \(a^{n}\),
  such that \(r_1r_2^nr_3r_4^nr_5\) accepts for all \(n\).

  Surely, we must have it that \(r_2\) reads at least one \(a\) from the first block, and \(r_4\) an
  \(a\) from the second block: At least one \(a\) is read, and if no \(a\) from the second block were read, we could loop
  \(r_2, r_4\) to break the equinumerosity. Similarly, we must also
  have it that \(r_2\) pushes a value to the stack which \(r_4\) pops, otherwise
  they could be repeated independently.
  This gives the \(r\)-coupling, and thus the \(A\)-coupling.
\end{proof}

\begin{lemma}\label{lem:notcfl}
  If \(\pkh(\hard(A)) \neq \emptyset\), then \(a^nb^{*}a^nb^* \shuf_T d^*c^nd^{*}c^{n}\) is
  not a CFL.
\end{lemma}

\begin{proof}
  Let \(L_1 = a^nb^{*}a^nb^*\), and let \(L_2 = d^*c^nd^{*}c^{n}\).
By \Cref{lem:anbna}, there
   is some \(k_1\) such that any big enough word \(w \in L_1\) has the first \(a\)'s
   coupled with the last \(a\)'s.
   Similarly, there is a constant \(k_{2}\) for \(L_{2}\) and the \(c\)'s.
   Take \(k=\max(k_{1},k_{2})\).

  Fix an SCC-pattern \(P\), and suppose first that \(P\) contains an \(st\)-SCC \(S\)
  with \(n_s\), \(n_t\) letters of \(s, t\) for some path, both nonzero, with a
  ratio of \(r = \frac{n_s}{n_t}\).
  We assume \(r\geq1\), the other case being symmetric.
  Consider the shuffle of \(w_1 =a^{rk}b^{rk}a^{rk}b^{rk}\) and \(w_2 = d^{k}c^{k}d^{k}c^{k} \), for \(k\) a multiple of \(n_{t}\).
  By construction, \(w_{1}\shuf_{S} w_{2} = (a^{rk}\shuf_{S}d^{k})(b^{rk}\shuf_{S}c^{k})(a^{rk}\shuf_{S}d^{k})(b^{rk}\shuf_{S}c^{k})\).
  In the shuffled word, the first part is coupled with the third part, and the second part with the last part.
  The couplings cross. We conclude,
  noting that the \(s, t, s, t\) case is comparatively trivial.
\end{proof}

\begin{lemma}\label{lem:noantigrid}
  If \(\pkh(\hard(A))\) has no antigrid, then \(\shuf_T\) is \dcflp_hostile.
\end{lemma}

\begin{proof}
  We first show that this condition implies that \(\pi_1(L_1 \shuf_T L_2)\)
  is nonregular for any nonregular \(L_1, L_2\).  Suppose \(L_1\) is in
  \(\dcflp\), and \(L_2\) an infinite language in \(\dcfl\).  By \Cref{lem:js},
  \(L_1\) has a sublanguage of the form \(\{u_1x_1^nv_1y_1^nz_1 | n \geq 0\}\).
  The Parikh image of this sublanguage \(L_1\) is thus a periodic set \(U\) with
  base \(|u_1v_1z_1|\) and period \(|x_1y_1|\). By Parikh's
  theorem~\cite{Parikh66}, \(\pkh(L_2)\) similarly contains a periodic set
  \(V\). As \(\pkh(\hard(A))\) has no antigrid, these two must intersect
  infinitely often under \(T\), and \(\pi_1(L)\) must be nonregular for \(L =
  L_1 \shuf_T L_2\). The same symmetrically holds for \(\pi_2(L)\), and hence,
  shuffle resilience is available for both.

  Now, suppose \(L_1, L_2\) are both \(\dcflp\). By \Cref{lem:detcoupres}, there
  is a sublanguage of \(L_1\) of the same form such that the \(i\)-th occurrence
  of \(x_1\) has a shuffle resilient coupling with the \(n-i+1\)-th of \(y_1\),
  with an analogous \(\{u_2x_2^nv_2y_2^nz_2 | n \geq 0\}\) for \(L_2\).  By
  taking bigger blocks, we can assume that \(|x_{1}|=|x_{2}|=|y_{1}|=|y_{2}|\).

  Let \(k\) be a constant such that every run of length greater than \(k\) must
  repeat the same simple loop at least \(|A|!\) times.  Such a constant clearly
  exists by the pigeonhole principle.

  We say that an SCC-pattern is \emph{synchronized} if every loop \(l\) on the
  path has the same ratio \(\alpha = \frac{|l|_s}{|l|_t}\).  Let \(P=
  r_1P_1r_2\ldots r_nP_nr_{n+1} \) be a synchronized pattern.  Let \(\delta\) be
  the difference between the number of \(s\) and \(t\) in \(r_1\ldots r_{n+1}\).
  It is clear that every word \(w\) accepted by \(P\) is such that \(|w|_{s} =
  \alpha |w_{t}|+\delta\).  That is, its Parikh image lies on a line.  Consider
  the subgrid of \((|u_1x_1v_1y_1z_1| + |x_1y_1|\mathbb{N}) \times
  (|u_2x_2v_2y_2z_2| + |x_2y_2| \mathbb{N})\) avoiding all lines for the
  synchronized patterns of \(A\).  This can be constructed following the proof
  of~\cref{prop:decide_antigrids}.  By rebasing the grid, we can assume that
  \(n\) and \(m\) are as large as we want.  Because \(\pkh(\hard(A))\) has no
  antigrid, there are \(n,m\geq 9k\) such that \(w_{1}=u_1x_1^nv_1y_1^nz_1\) can
  be shuffled with  \(w_{2}=u_2x_2^mv_2y_2^mz_2\).  Let \(P\) be the pattern for
  which the shuffle it possible, with a run \(r\) in \(A\).  By construction,
  \(P\) cannot be a synchronized pattern. We treat two cases.

  First assume that \(P\) contains an \(s\)-SCC, \(t\)-SCC, \(s\)-SCC, and
  \(t\)-SCC: Label these as the \(B_1 \ldots B_4\) blocks.  Let \(\lambda_1\)
  (resp. \(\lambda_2\)) be the length of some simple loop in \(B_{1}\) (resp.
  \(B_{4}\)).  We decompose \(r =
  r_{1}\rho_{1}r_{2}\rho_{2}r_{3}\rho_{3}r_{4}\rho_{4}r_{5}\) where each
  \(\rho_{i}\) is the part of the run in \(B_{i}\).  If \(|r_{i}| > k\) for some
  \(i\), then some loop is repeated in \(r_{i}\) at least \(\lambda_1\lambda_2\)
  times, accounting for some \(l\lambda_1\lambda_2\) number of \(s\) or \(t\).
  We can remove these loops and add the corresponding number of loops in
  \(B_{1}\) or \(B_{4}\) to have another run with the same Parikh image.  By
  balancing loops between \(B_{1}\) and \(B_{3}\), and \(B_{2}\) and \(B_{4}\),
  we can assume \(|\rho_{i}| \geq k\).  With this run, the \((k+1)\)-th
  occurrence of \(x_1\) is shuffled by \(B_{1}\), and the \((n-k+1)\)-th
  occurrence of \(y_{1}\) by \(B_{3}\).  Similarly, the \((k+1)\)-th occurrence
  of \(x_2\) is shuffled by \(B_{2}\), and the \((m-k+1)\)-th occurrence of
  \(y_{2}\) by \(B_{4}\).  This gives a crossing.

  As for the case where \(P\) contains an \(st\)-SCC: it is not synchronized
  hence there are two loops \(l_{1}\) and \(l_{2}\) with respective, nonequal
  ratios \(\alpha = |l_1|_s/|l_1|_t\) and \(\beta = |l_2|_s/|l_2|_t\); assume
  \(\alpha < \beta\).  Denote by \(B_{1}\) and \(B_{2}\) the SCCs of these two
  loops, and assume that \(B_{1}\) is before \(B_{2}\) on the path. Then, we
  have a decomposition \(r=r_1\rho_1r_2 \rho_2r_3\), where \(\rho_1, \rho_2\)
  read \(l_1^{n_1}, l_2^{n_2}\) for some large \(n_1, n_2\). Note that iterating
  these loops by any multiple of \(|x_1y_1|\) results in a valid word in the
  shuffle.

  Let \(L\) be the shared length of the blocks. In \(\rho_1\), the \(i\)-th
  block of \(w_1\) and \(j\)-th block of \(w_2\) are at positions \(|r_{1}| + i
  L (1 + \alpha^{-1})\) and \(|r_{1}| + j L (1 + \alpha)\) respectively.
  Similarly, in \(\rho_2\), their positions relative to the end of the run are
  \(-|r_{3}| - i L (1 + \beta^{-1})\) and \(-|r_{3}| - j L (1 + \beta)\).  A
  crossing occurs if we can find indices \(i, j\) such that the \(j\)-th block
  of \(x_2\) ends before the \(i\)-th block of \(x_1\) begins (\(x_2 < x_1\)),
  but the \(j\)-th block of \(y_2\) ends before the \(i\)-th block of \(y_1\)
  begins (\(y_2 < y_1\)).  The condition \(x_2 < x_1\) translates to
  \((j+1)L(1+\alpha) < iL(1+\alpha^{-1})\), which simplifies to \(\alpha(j+1) <
  i\).  The condition \(y_2 < y_1\) means \(y_2\) is further from the end than
  \(y_1\); in terms of offsets from the run's end, the end of \(y_2\) must be at
  a larger offset than the start of \(y_1\).  This gives \((j-1)L(1+\beta) >
  iL(1+\beta^{-1})\), which simplifies to \(i < \beta(j-1)\).  Combining these,
  we require \(\alpha(j+1) < i < \beta(j-1)\).  As \(\alpha < \beta\), the
  interval \((\alpha(j+1), \beta(j-1))\) grows linearly with \(j\), ensuring
  that such integers \(i, j\) exist for large \(j\).  Choosing \(n_1, n_2\) to
  accommodate these indices yields the desired crossing.
\end{proof}

\begin{lemma}\label{lem:antigrid}
  If \(\pkh(\hard(A))\) has an antigrid, then there are \(L, M\) nonregular DCFLs
  such that \(L \shuf_T M\) is a CFL.
\end{lemma}
\begin{proof}
  Assume \((i + p\bbN) \times (j + q \bbN)\) is a grid in the complement of
  \(\pkh(\hard(A))\). Let \(L_1\) be any \dcflp language of words of length
  \((i + p\bbN)\) and similarly \(L_2\) any \dcflp language of words with lengths in
  \((j + q\bbN)\).  Clearly, \(L_1 \shuf_T L_2 = L_1 \shuf_{\easy(A)} L_2\), and the
  \Cref{lem:safe} concludes.
\end{proof}

\begin{remark}\label{rem:complexity2}
  The decision algorithm of \Cref{thm:dcflhost} is overall \(\TWOEXPSPACE\):
  Computing a formula for the Parikh image from an automaton yields a doubly
  exponential blowup~\cite{KW10}, which yields the above when applied to the
  complexity noted in~\cref{rem:complexity}.
\end{remark}

\section{Conclusion and Perspectives}

We have examined the conditions under which shuffling two nonregular CFL along a
regular trajectory results in another CFL. In doing so, we investigated the
notion of \emph{couplings}, i.e., portions of words where a PDA has to push and
pop respectively. We demonstrated that these couplings must exist (Coupling
Lemmas), that they can survive being shuffled (Resilience Lemmas), and applied
these tools to the study of trajectories of shuffles. There, we demonstrate that
some trajectories never shuffle two nonCFLs into CFLs, and further characterized
the behavior of said trajectories with respect to shuffling DCFLs.

The most pressing question left open by this work is to extend this latter
characterization to CFLs.  Bérard~\cite{berard87} hints at the fact that this
will be much harder: Indeed, \((st)^*(s^*+t^*)\) can shuffle some nonregular
CFLs into a CFL, and it is not known if the same holds for
\((s^*+t^*)(st)^*(s^*+t^*)\).  We do not venture to conjecture one way or
another. Similarly, this question is mirrored in the general coupling lemma:
does there exist a stronger description of when couplings do (or do not) exist
in the general case?

Going further than shuffles along regular trajectories, one could study
\emph{context-free} trajectories, an avenue of research started by Mateescu et
al.~\cite{mateescurs98}.  Further, (trajectory-based) variants of the so-called
\emph{synchronized} shuffle~\cite{BeekMM05}, in which some letters act as
synchronizations, ought to be studied in the same way, in particular given their
relevance in molecular biology.  (See also~\cite{BeekK09} for an extensive list
of references pertaining to this shuffle.)

\bibliographystyle{plainurl}
\bibliography{bib}

@Book{hopcroft79,
  author =	"John Hopcroft and Jeffrey Ullman",
  title =	"Introduction to Automata Theory, Languages, and
		 Computation",
  publisher =	"Addison Wesley",
  year = 	"1979",
}

@Book{sipser97,
  author =	"Michael Sipser",
  title =	"Introduction to the Theory of Computation",
  publisher =	"PWS Publishing Co.",
  year = 	"1997",
  address =	"Boston, Massachusetts",
  ISBN = 	"0-534-944728-X",
}

@book{baeten90,
  author       = {Jos C. M. Baeten and
                  W. P. Weijland},
  title        = {Process algebra},
  series       = {Cambridge tracts in theoretical computer science},
  volume       = {18},
  publisher    = {Cambridge University Press},
  year         = {1990},
  isbn         = {978-0-521-40043-5},
  bibsource    = {dblp computer science bibliography, https://dblp.org}
}

@article{bergstrak84,
  author       = {Jan A. Bergstra and
                  Jan Willem Klop},
  title        = {Process Algebra for Synchronous Communication},
  journal      = {Inf. Control.},
  volume       = {60},
  number       = {1-3},
  pages        = {109--137},
  year         = {1984},
  url          = {https://doi.org/10.1016/S0019-9958(84)80025-X},
  doi          = {10.1016/S0019-9958(84)80025-X},
  bibsource    = {dblp computer science bibliography, https://dblp.org}
}

@article{salomaay99,
  author       = {Kai Salomaa and
                  Sheng Yu},
  title        = {Synchronization Expressions and Languages},
  journal      = {J. Univers. Comput. Sci.},
  volume       = {5},
  number       = {9},
  pages        = {610--621},
  year         = {1999},
  url          = {https://doi.org/10.3217/jucs-005-09-0610},
  doi          = {10.3217/JUCS-005-09-0610},
  bibsource    = {dblp computer science bibliography, https://dblp.org}
}

@inproceedings{iwama83,
  author       = {Kazuo Iwama},
  editor       = {David S. Johnson and
                  Ronald Fagin and
                  Michael L. Fredman and
                  David Harel and
                  Richard M. Karp and
                  Nancy A. Lynch and
                  Christos H. Papadimitriou and
                  Ronald L. Rivest and
                  Walter L. Ruzzo and
                  Joel I. Seiferas},
  title        = {Unique Decomposability of Shuffled Strings: {A} Formal Treatment of
                  Asynchronous Time-Multiplexed Communication},
  booktitle    = {Proceedings of the 15th Annual {ACM} Symposium on Theory of Computing,
                  25-27 April, 1983, Boston, Massachusetts, {USA}},
  pages        = {374--381},
  publisher    = {{ACM}},
  year         = {1983},
  url          = {https://doi.org/10.1145/800061.808768},
  doi          = {10.1145/800061.808768},
  bibsource    = {dblp computer science bibliography, https://dblp.org}
}

@Inbook{nivat1982,
author="Nivat, Maurice",
title="Behaviors of Processes and Synchronized Systems of Processes",
bookTitle="Theoretical Foundations of Programming Methodology: Lecture Notes of an International Summer School, directed by F. L. Bauer, E. W. Dijkstra and C. A. R. Hoare",
year="1982",
publisher="Springer Netherlands",
address="Dordrecht",
pages="473--551",
isbn="978-94-009-7893-5",
doi="10.1007/978-94-009-7893-5_14",
url="https://doi.org/10.1007/978-94-009-7893-5_14"
}

@inproceedings{OgdenRR78,
  author       = {William F. Ogden and
                  William E. Riddle and
                  William C. Rounds},
  editor       = {Alfred V. Aho and
                  Stephen N. Zilles and
                  Thomas G. Szymanski},
  title        = {Complexity of Expressions Allowing Concurrency},
  booktitle    = {Conference Record of the Fifth Annual {ACM} Symposium on Principles
                  of Programming Languages, Tucson, Arizona, USA, January 1978},
  pages        = {185--194},
  publisher    = {{ACM} Press},
  year         = {1978},
  url          = {https://doi.org/10.1145/512760.512780},
  doi          = {10.1145/512760.512780},
  bibsource    = {dblp computer science bibliography, https://dblp.org}
}

@article{Shaw78,
  author       = {Alan C. Shaw},
  title        = {Software Descriptions with Flow Expressions},
  journal      = {{IEEE} Trans. Software Eng.},
  volume       = {4},
  number       = {3},
  pages        = {242--254},
  year         = {1978},
  url          = {https://doi.org/10.1109/TSE.1978.231501},
  doi          = {10.1109/TSE.1978.231501},
  bibsource    = {dblp computer science bibliography, https://dblp.org}
}

@article{BeekMM05,
  author       = {Maurice H. ter Beek and
                  Carlos Mart{\'{\i}}n{-}Vide and
                  Victor Mitrana},
  title        = {Synchronized shuffles},
  journal      = {Theor. Comput. Sci.},
  volume       = {341},
  number       = {1-3},
  pages        = {263--275},
  year         = {2005},
  url          = {https://doi.org/10.1016/j.tcs.2005.04.007},
  doi          = {10.1016/J.TCS.2005.04.007},
  bibsource    = {dblp computer science bibliography, https://dblp.org}
}

@article{BeekK09,
  author       = {Maurice H. ter Beek and
                  Jetty Kleijn},
  title        = {Associativity of Infinite Synchronized Shuffles and Team Automata},
  journal      = {Fundam. Informaticae},
  volume       = {91},
  number       = {3-4},
  pages        = {437--461},
  year         = {2009},
  url          = {https://doi.org/10.3233/FI-2009-0051},
  doi          = {10.3233/FI-2009-0051},
  bibsource    = {dblp computer science bibliography, https://dblp.org}
}

@article{berard87,
  author       = {B{\'{e}}atrice B{\'{e}}rard},
  title        = {Literal Shuffle},
  journal      = {Theor. Comput. Sci.},
  volume       = {51},
  pages        = {281--299},
  year         = {1987},
  url          = {https://doi.org/10.1016/0304-3975(87)90037-5},
  doi          = {10.1016/0304-3975(87)90037-5},
  bibsource    = {dblp computer science bibliography, https://dblp.org}
}

@article{mateescurs98,
  author       = {Alexandru Mateescu and
                  Grzegorz Rozenberg and
                  Arto Salomaa},
  title        = {Shuffle on Trajectories: Syntactic Constraints},
  journal      = {Theor. Comput. Sci.},
  volume       = {197},
  number       = {1-2},
  pages        = {1--56},
  year         = {1998},
  url          = {https://doi.org/10.1016/S0304-3975(97)00163-1},
  doi          = {10.1016/S0304-3975(97)00163-1},
  bibsource    = {dblp computer science bibliography, https://dblp.org}
}

@article{latteux79,
  author       = {Michel Latteux},
  title        = {C{\^{o}}nes rationnels commutatifs},
  journal      = {J. Comput. Syst. Sci.},
  volume       = {18},
  number       = {3},
  pages        = {307--333},
  year         = {1979},
  url          = {https://doi.org/10.1016/0022-0000(79)90039-4},
  doi          = {10.1016/0022-0000(79)90039-4},
  bibsource    = {dblp computer science bibliography, https://dblp.org}
}

@inproceedings{jancars21,
  author       = {Petr Jancar and
                  Jiv{r}'{i} v{S}'{i}ma},
  editor       = {Filippo Bonchi and
                  Simon J. Puglisi},
  title        = {The Simplest Non-Regular Deterministic Context-Free Language},
  booktitle    = {46th International Symposium on Mathematical Foundations of Computer
                  Science, {MFCS} 2021, August 23-27, 2021, Tallinn, Estonia},
  series       = {LIPIcs},
  volume       = {202},
  pages        = {63:1--63:18},
  publisher    = {Schloss Dagstuhl - Leibniz-Zentrum f{\"{u}}r Informatik},
  year         = {2021},
  url          = {https://doi.org/10.4230/LIPIcs.MFCS.2021.63},
  doi          = {10.4230/LIPICS.MFCS.2021.63},
  bibsource    = {dblp computer science bibliography, https://dblp.org}
}

@article{BordihnHK05,
  author       = {Henning Bordihn and
                  Markus Holzer and
                  Martin Kutrib},
  title        = {Unsolvability levels of operation problems for subclasses of context-free
                  languages},
  journal      = {Int. J. Found. Comput. Sci.},
  volume       = {16},
  number       = {3},
  pages        = {423--440},
  year         = {2005},
  url          = {https://doi.org/10.1142/S0129054105003078},
  doi          = {10.1142/S0129054105003078},
  bibsource    = {dblp computer science bibliography, https://dblp.org}
}

@article{BordihnM20,
  author       = {Henning Bordihn and
                  Victor Mitrana},
  title        = {On the degrees of non-regularity and non-context-freeness},
  journal      = {J. Comput. Syst. Sci.},
  volume       = {108},
  pages        = {104--117},
  year         = {2020},
  url          = {https://doi.org/10.1016/j.jcss.2019.09.003},
  doi          = {10.1016/J.JCSS.2019.09.003},
  bibsource    = {dblp computer science bibliography, https://dblp.org}
}

@article{Stearns67,
  author       = {Richard Edwin Stearns},
  title        = {A Regularity Test for Pushdown Machines},
  journal      = {Inf. Control.},
  volume       = {11},
  number       = {3},
  pages        = {323--340},
  year         = {1967},
  url          = {https://doi.org/10.1016/S0019-9958(67)90591-8},
  doi          = {10.1016/S0019-9958(67)90591-8},
  bibsource    = {dblp computer science bibliography, https://dblp.org}
}

@INPROCEEDINGS{KW10,
  author={Kopczynski, Eryk and To, Anthony Widjaja},
  booktitle={2010 25th Annual IEEE Symposium on Logic in Computer Science}, 
  title={Parikh Images of Grammars: Complexity and Applications}, 
  year={2010},
  volume={},
  number={},
  pages={80-89},
  doi={10.1109/LICS.2010.21}}

@article{Parikh66,
  author       = {Rohit Parikh},
  title        = {On Context-Free Languages},
  journal      = {J. {ACM}},
  volume       = {13},
  number       = {4},
  pages        = {570--581},
  year         = {1966},
  url          = {https://doi.org/10.1145/321356.321364},
  doi          = {10.1145/321356.321364},
  bibsource    = {dblp computer science bibliography, https://dblp.org}
}

@article{nguyen2018polyhedra,
  title = {Enumerating Projections of Integer Points in Unbounded Polyhedra},
  author = {Nguyen, Danny and Pak, Igor},
  journal = {SIAM Journal on Discrete Mathematics},
  volume = {32},
  number = {2},
  pages = {986-1002},
  year = {2018},
  doi = {10.1137/17M1118907}
}

\ifbool{for_arxiv}{
    \newpage
\appendix

\section{Proofs of \Cref{sec:couppda}}

\FactEndcap*
\begin{proof}
  The first two properties are immediate by construction.  For the last one, let
  \(e \in \Gamma^+\) be such that \(\stk(r, j) = \stk(r, i)e\); we show that all the
  symbols of \(e\) were pushed by \(r\ang{x}\).  Using the notations of the
  definition, we must have that \(|\stk(r, i)| = h_x\), and so the symbols of
  \(e\) are not pushed by \(r\ange{u[i+1:]}\).  Further, consider the rightmost
  \(t \in r\ang{x}\) with \(\stk(r,i) = \stk(r,t)\) and let
  \(k = \max r\ang{x}\).  If \(\stk(r, t) = \stk(r, k)\), then \(h_v\), the lowest
  stack-height in \(r\ange{v}\), is achieved at the boundary between
  \(r\ang{x}\) and \(r\ange{v}\), so that \(k = j\) and
  \(\stk(r, i) = \stk(r, j)\), which is not possible.  Hence
  \(\stk(r, t) = \stk(r, i)e'\) with a nonempty \(e' \in \Gamma^+\) pushed while reading
  \(x\).  The lowest stack-height in \(r\ange{v}\) is first achieved either at the
  boundary, or when a transition of \(r\ange{v}\) pops a symbol of \(e'\).  In both
  cases, when reaching \(j\), \(e\) contains a nonempty prefix of \(e'\).
\end{proof}

\LemFinClass*
\begin{proof}
  That it is an equivalence relation is immediate by definition.  Let
  \(e \in \Gamma^*\) and define the sets:
  \begin{align*}
    Q_e^= & = \{(q, w, \stt(r)) \mid q \in Q, w \in \Sigma^{\leq k}, \text{ and \(r\) is a run
            for \(w\) in \(A_{e,q}\) with \(\stk(r) = \eps\)}\}\\
    Q_e^> & = \{(q, w, \stt(r)) \mid q \in Q, w \in \Sigma^{\leq k}, \text{ and \(r\) is a run
            for \(w\) in \(A_{e,q}\) with \(\stk(r) \neq \eps\)}\}\\
  \end{align*}
  Naturally, since \(k\) is fixed, there are finitely many such sets for all
  values of \(e\).  We show that if \(e, e' \in \Gamma^*\) are such that
  \((Q_e^=, Q_e^>) = (Q_{e'}^=, Q_{e'}^>)\), then they are
  \(k\)_suffix-equivalent, proving the claim.  Let \(s \in \Gamma^*, q \in Q\), and assume
  \(A_{se,q}\) accepts \(w \in \Sigma^{\leq k}\) with some run \(r\).  We show that
  \(A_{se', q}\) also accepts \(w\).  We distinguish two cases:
  \begin{enumerate}[({Case} 1.)\quad,wide,font=\it]
  \item If \(\stk(r, i) = s\) at some point, consider the smallest such
    \(i\).  The subrun \(r[:i]\) is thus a valid run in \(A_{e,q}\) to some state
    \(q'\) with a label \(w'\) of length smaller than \(k\).  When seen as a run in
    \(A_{e,q}\), we have that \(\stk(r[:i]) = \eps\).  Consequently, \((q, w', q') \in
    Q_e^=\).  Since \(Q_e^= = Q_{e'}^>\), there is a run \(r'\) in \(A_{e',q}\) to the
    same state \(q'\) and with that same label \(w'\), with \(\stk(r') = \eps\).
    Thus, the path \(r'r[i+1:]\) is an accepting path in \(A_{se',q}\) with label
    \(w\).
  \item If \(|\stk(r, i)| > |s|\) for all \(i\), then \(r\) is also an accepting run
    in \(A_{e,q}\), and thus \((q, w, \stt(r)) \in Q_e^>\).  Since \(Q_e^> = Q_{e'}^>\),
    there is an accepting run in \(A_{e',q}\) for \(w\), and it is also an accepting
    run in \(A_{se',q}\).\qedhere
  \end{enumerate}
\end{proof}

\LemSurgery*
\begin{proof}
  Let the min-endcap of \(r_1\) be between \(i_1\) and \(j_1\), and that of \(r_2\)
  between \(i_2\) and \(j_2\).  Let \(r = r_1r_yr_z\) be an accepting run, with \(z\) the
  label of \(r_z\) such that \(|z| \leq k\).  We distinguish two cases:
  \begin{enumerate}[({Case} 1.)\quad,wide,font=\it,ref=\arabic*]
  \item If the stack-effects of the
    min-endcap are empty, then the path:
    \[r_{121} = r_1[:i_1]r_2[i_2+1:j_2]r_1[j_1+1:]\]%
    is valid and equivalent, in stack and state reached, to \(r_1\).  This is
    because \(r_2[i_2+1:j_2]\) never pops below the starting state of its stack,
    and adds nothing to the stack.  Additionally, \(r_{121}\) is labeled
    \(w_2 = u_2x_2v_2\).  Hence \(r_{121}r_yr_z\) is an accepting run with label
    \(w_2yz\).
  \item\label{c} Otherwise, \Cref{fact:endcap} indicates that \(\stk(r_1, j_1) = \stk(r_1,
    i_1)e_1\), with \(e_1\), the stack-effect of the min-endcap of \(r_1\), a nonempty
    sequence pushed by \(r\ang{x}\).  Let \(t\) be the first position in \(r\ange{z}
    \setminus r\ange{y}\) such that \(\stk(r, j_1) = \stk(r, t)\); we distinguish two
    subcases:
    \begin{enumerate}[({Case \ref{c}.}1.)\quad,wide,font=\it]
    \item If no such \(t\) exists, then the run:
      \[r_1[:i_1]r_2[i_2+1:j_2]r_1[j_1+1:]r_yr_z\]
      is valid and accepting, with label \(w_2yz\).  This is because the subrun
      \(r_1[j_1+1:]r_yr_z\) never pops below the starting state of its stack.
    \item If such a \(t\) exists, and seeing \(t\) as a position in \(r_z\), consider the runs:
      \begin{align*}
        r_1' & = r_1[:i_1]r_1[i_1+1:j_1]r_1[j_1+1:]r_yr_z[:t] = r_1r_yr_z[:t]\\
        r_2' & = r_1[:i_1]r_2[i_2+1:j_2]r_1[j_1+1:]r_yr_z[:t].
      \end{align*}
      Note that \(r_2[i_2+1:j_2]\) has stack-effect \(e_2\) and never pops below the
      starting state of its stack.  In addition, \(r_1[j_1+1:]r_yr_z[:t]\) leaves
      the stack unchanged and never pops below the starting state of its stack
      (since \(r_1\) does not and neither does \(r_y\) since \(x\) and \(y\) are not
      coupled).  This means that, writing \(s = \stk(r_1, i_1)\), we have
      \(\stk(r_1') = se_1\) and \(\stk(r_2') = se_2\) and both end in the same state
      \(q\).  Let \(z'\) be the label of \(r_z[t+1:]\).  Since \(e_1\) and
      \(e_2\) are \(k\)_suffix-equivalent and \(|z'| \leq k\), and since
      \(A_{se_1, q}\) accepts \(z'\) by following \(r_z[t+1:]\), so does
      \(A_{se_2}\).  Finally, we note that the label of \(r_2'\) concatenated with
      \(z'\) is precisely \(w_2yz\), showing the claim.\qedhere
    \end{enumerate}
  \end{enumerate}
\end{proof}

\section{Proofs of \Cref{sec:coupdpda}}

\LemUnbalancedOgden*

\begin{proof}
 As a point of technique, pumping lemmata are generally much simpler to prove
 using context-free grammars instead of PDAs. Thus, we adopt this formalism for
 this proof only, noting that we only make a statement about the language
 itself. We refer the reader to~\cite{hopcroft79} for the basic definitions and
 facts about context-free grammars (CFG).

  Take a CFG \(G\) for \(L\), with \(N\) as set of non-terminals and \(S\) as starting non-terminal.
  Since \(L\) does not contain the empty word, we can assume it is in Chomsky reduced form.
  That is, all productions are of the form \(X\rightarrow YZ\) or \(X\rightarrow a\) with \(X,Y,Z\in N\) and \(a\in \Sigma\).

  Set \(p=2^{|N|+1}\) for the rest of the proof.
  We consider derivation trees \(t\) for words of the form \(w_{1}\#w_{2}\in L\).
  The subword left (resp. right) of \(\#\) is referred as the left part (resp. right part).
  A \emph{spot} is a pair of nodes \(s=(x,y)\) such that \(x\) is an ancestor of \(y\).
  The subtree with \(x\) as root and \(y\) identified as a leaf is called \(\st(x,y)\).
  The derivation tree with \(x\) and \(y\) merged, and the merged node labeled by \(\st(x,y)\) is called \(\excl(t;x,y)\)
  We can define the \emph{left yield} \(g_{s}\) (resp. \emph{right yield} \(r_{s}\)) as the word formed by the leaves left (resp. right) of \(y\) in that subtree.
  We say that \(s\) is a \emph{pumping spot} if \(x\) and \(y\) are labeled by the same non-terminal.
  It is further \emph{biased} if either: \(g_{s}\) is in the left part, \(r_{s}\) is in the right part and \(p\geq |g_{s}|/ |\lambda| > |r_{s}|/|\mu|\); or both \(g_{s}\) and \(r_{s}\) are in the \(\lambda\)-part and \(p\geq |g_{s}|+|r_{s}|\).
  By applying a pumping spot, we mean adding an occurrence of \(\st(x,y)\) at \(s\); which is possible because \(x\) and \(y\) are labeled by the same non-terminal.

  We prove the following by induction on \(|w_{1}|+|w_{2}|\).
  For any derivation tree \(t\) for \(w=w_{1}\#w_{2}\in L\) such that \(|\mu||w_{1}|>|\lambda||w_{2}|\), there is a derivation tree \(t'\) (for some other word) with a biased pumping spot \(s\) such that \(\excl(t';s)=t\).
  The base case \(|w_{1}|+|w_{2}|=0\) is always true as there exists no such derivation tree.
  We go on to the general case.
  We start by applying the same technique as in the classical pumping lemma.
  Because \(|w|>2^{|N|}\), there is a path of length at least \(|N|+1\) in the tree.
  Thus we can find a repetition of a non-terminal at most at distance \(|N|+1\) of the leaf of the path.
  This give a pumping spot \(s=(x,y)\) with \(0<|g_{s}|+|r_{s}|\leq 2^{|N|+1}\).
  There is exactly one symbol \(\#\) in every word of \(L\), hence none of \(g_{s}\) and \(r_{s}\) contain \(\#\) and they are both included in either part of \(w\).
  There are three cases.
  \begin{itemize}
    \item If both \(g_{s}\) and \(r_{s}\) are in the left part. Then, \(s\) is a
    biased pumping spot according to the second part of the definition.
    Applying \(s\) once leads to the conclusion.

    \item If both \(g_{s}\) and \(r_{s}\) are in the right part. Then, \(\excl(t;x,y)\) is a derivation tree for a word \(w'=w_{1}\# w'_{2}\) with \(|w'_{2}|<|w_{2}|\).
          Let \(z\) be the node in \(\excl(t;x,y)\) that corresponds to the merging of \(x\) and \(y\).
          Thus we can apply the induction hypothesis to obtain a derivation tree \(t'\) with a biased pumping spot \(s\) such that \(\excl(t';s)=\excl(t;x,y)\).
          Thanks to that, we can find in \(\excl(t';s)\) a copy of \(z\) that can be replaced by \(\st(x,y)\).
          The new derivation tree is exactly \(t\), and so, putting back
          \(\st(s)\) completes the construction.
    \item If \(g_{s}\) is in the left part and not \(r_{s}\).
          If further \(|g_{s}|/ |\lambda| > |r_{s}|/|\mu|\), \(s\) is a biased pumping spot according to the first part of the definition. Applying \(s\) once concludes.
          In the other case, we consider \(\excl(x,y)\) and the word \(w'=w'_{1}\# w'_{2}\) it derives.
          We have, \(\frac{|w'_{1}|}{|\lambda|} = \frac{|w_{1}|}{|\lambda|} -
          \frac{|g_{s}|}{|\lambda|} > \frac{|w_{2}|}{|\mu|} -
          \frac{|r_{s}|}{|\mu|} = \frac{|w'_{2}|}{|\mu|}\).
          Thus we can apply the induction hypothesis to obtain a derivation tree \(t'\) with a biased pumping spot \(s\) such that \(\excl(t';s)=\excl(t;x,y)\).
          Thanks to that, we can find in \(\excl(t';s)\) a copy of \(z\) that can be replaced by \(\st(x,y)\).
          The new derivation tree is exactly \(t\), and so, putting back
          \(\st(s)\) completes the construction.
  \end{itemize}

  Take any \(w=w_{1}\# w_{2}=\lambda'\lambda^{n}\#\mu^{n'}\mu'\in L\) with \(n\geq p\) and \(n>n'+1\), and any derivation tree \(t\) for \(w\).
  We evaluate \(|w_{1}| = |\lambda'|+n|\lambda|\geq n|\lambda|\) and \(|w_{2}|=|\mu'|+n'|\mu|\leq (n'+1)|\mu|\), giving \(|\mu||w_{1}|>|\lambda||w_{2}|\).
  Thus, with the result with proved by induction, there is a derivation tree \(f'\) for \(w'=w'_{1}\# w'_{2}\) with a biased pumping spot \(s\) and \(\excl(t';s)=f\).
  If \(r_{s}\) is in the left part, let \(k'=|g_{s}|+|r_{s}|\) and \(l'=0\); otherwise let \(k'=|g_{s}|\) and \(l'=|r_{s}|\).
  For \(m\geq 1\), applying the pumping \(m-1\) times gives a word \(w''=w''_{1}\# w''_{2}\in L\) with \(|w''_{1}|=|w_{1}| + k'm\) and \(|w''_{2}|=|w_{2}| + l'm\).
  However, there is a unique word in \(L\) with these length conditions: there are \(k,l\) such that \(k'= k|\lambda|\) and \(l'=l|\mu|\), and \(w''=\lambda'\lambda^{n+km}\# \mu^{n'+lm}\mu'\).
  Finally, for either definitions of \(k'\) and \(l'\), we have \(k>l\).
\end{proof}

\LemPumpingTwo*
\begin{proof}
  Let \(r\) be such a run, and consider the states encountered in \(r\) following the given decomposition, where \(\varepsilon\)-transitions are assigned arbitrarily to a part.
  That is we let \(p_{1} = \stt(r\ange{u})\), \(p_{2}= \stt(r\ang{ \lambda'\lambda^{n}})\), \(p_{3} = \stt(r\ange{v})\) and \(p_{4}= \stt(r\ang{\mu^{n'} \mu'})\).
  Let \(A_{1}\) be the PDA defined as \(A\) with no final state and \(p_{1}\) being the only initial state.
  Let \(A_{2}\) be the PDA defined as \(A\) with no initial state and \(p_{4}\) being the only final state.
  Let \(A'\) be the automaton that consists of disjoint copies of \(A_{1}\) and \(A_{2}\) with a transition \(t\) from \(p_{2}\) to \(p_{3}\) labeled by a fresh alphabet symbol \(\#\) and leaving the stack untouched.
  The usual requirement that every transition has a stack action can be met via matching push–pop operations.
  We require that the transition \(t\) and acceptance can only be triggered if the last transition used is \emph{not} an \(\varepsilon\)-transition.
  We add some conditions into \(A'\), depending on \(r\).
  \begin{itemize}
    \item If \(\stk(r\ange{u}) = \stk(r\ang{\lambda^{n}})\): the transition \(t\) can only be used if the stack is empty.
    \item If \(\stk(r\ang{ \lambda'\lambda^{n}}) \neq \stk(r\ange{v}):\) taking the transition \(t\) empties the stack.
    \item If \(\stk(r\ang{ \lambda'\lambda^{n}}) \neq \stk(r\ange{v})\) and \(\stk(r\ange{v}) = \stk(r\ang{\mu^{n'} \mu'})\): \(p_{4}\) is accepting only if the stack is empty.
    \item If \(\stk(r\ang{ \lambda'\lambda^{n}}) = \stk(r\ange{v})\) and \(\stk(r\ange{u}) = \stk(r\ang{\mu^{n'} \mu'})\): \(p_{4}\) is accepting only if the stack is empty.
  \end{itemize}
  The PDA \(A'\) can be tweaked to take these modifications into account with standard techniques, yielding that the language \(L'\) of \(A'\) is a CFL.
  We will use Ogden's lemma on \(L' \cap  \lambda'\lambda^{*}\# \mu^{*} \mu'\), hence we first need to relate this language with \(L\).

  \begin{claim}
    \label{claim:subclaim_pumping2}
    We have that \(L'\cap  \lambda'\lambda^{*}\#\mu^{*} \mu'\) is a subset of \(\{ \lambda'\lambda^{i}\# \mu^{j} \mu' \ |\ u\lambda'\lambda^{i}v\mu^{j} \mu'z\in L\}\) containing \(  \lambda'\lambda^{n}\#\mu^{n'} \mu'\).
   \end{claim}
  \begin{proof}
    \renewcommand{\qedsymbol}{$\vartriangleleft$}
    We prove it in the case \(\lambda'=\mu'=0\) for simplicity of notation, the general
    case being similar.
            Let \(\lambda^{i}\# \mu^{j} \in L'\) accepted by a run \(s\).
            We will replace parts of \(r\) by parts of \(s\) to obtain an accepting run labeled by \(u\lambda^{i}v \mu^{j}z\) in \(A\).
            To do that, we identify every transition of \(A'\) not labeled by \(\#\) by the corresponding transition in \(A\).
            We replace \(r\ang{\lambda^{n}}\) by \(s\ang{\lambda^{i}}\) and  \(r\ang{\mu^{n'}}\) by \(s\ang{\mu^{j}}\), giving a new run \(r'\).
            Assuming that the stack semantic is never violated, it is clear that the obtained state decomposition still goes through \(p_{1}, p_{2}, p_{3}\) and \(p_{4}\) to end in an accepting state.
            We argue that iteratively, for every possible version of \(A'\).
            \begin{itemize}
              \item First, it is clear that the stack semantic is correct in \(r'\ang{u\lambda^{i}}\): \(s\ang{\lambda^{i}}\) is valid no matter the content of \(\stk(r\ange{u})\).
              \item If \(\stk(r\ange{u}) = \stk(r\ang{\lambda^{n}})\), then \(\stk(r'\ang{\lambda^{i}}) = \stk(r\ang{\lambda^{n}})\).
            Otherwise, while reading \(\lambda^{i}\), \(r\) cannot pop anything from \(\stk(r\ange{u})\) (as \(u\) and \(\lambda^{i}\) are not \(r\)-coupled) and thus \(r\) pushes an element that is not popped.
            For the same reason with \(\lambda^{i}\) and \(v\) that are not \(r\)-coupled, \(v\) cannot pop that element, and thus cannot pop anything from \(\stk(r\ange{u})\).
            This means that \(r\ange{v}\) is valid no matter the content of the stack.
            In both cases, the stack semantic of \(r'\ange{u\lambda^{i}v}\) is correct.
              \item If \(\stk(r\ang{\lambda^{n}}) \neq \stk(r\ange{v})\), then there is an element pushed but not popped by \(r\) while reading \(v\).
            As before, because \(v\) and \(\mu^{n'}\) are not \(r\)-coupled, this implies that nothing in \(r\ang{\mu^{n'}}\) is popped from before.
            The condition on \(A'\) that the stack is emptied before reading \(\mu^{j}\) gives that the replacement is correct.
            Otherwise, \(\stk(r'\ange{u\lambda^{i}v}) = \stk(r'\ange{u\lambda^{i}} \). Thus adding \(s\ang{\mu^{j}}\) is valid, because \(r'\ange{u\lambda^{i}\mu^{j}}\) is.
            In both cases, the stack semantic of \(r'\ang{u\lambda^{i}v\mu^{j}}\) is correct.
              \item If \(\stk(r\ang{\lambda^{n}}) \neq \stk(r\ange{v})\) and \(\stk(r\ange{v}) = \stk(r\ang{\mu^{n'}})\), then while reading \(z\), \(r\) can only pop elements that were pushed by \(v\).
                    Because the definition of \(A'\) imposes in that case that \(\mu^{j}\) leaves the stack intact in \(s\), \(r'\ange{z}\) is valid.
                    If \(\stk(r\ang{\lambda^{n}}) \neq \stk(r\ange{v})\) and \(\stk(r\ange{v}) \neq \stk(r\ang{\mu^{n'}})\), then \(r\ang{\mu^{n'}}\) pushes an element on the stack that is not popped.
                    Because \(\mu^{n'}\) and \(z\) are not \(r\)-coupled, \(z\) only pops elements that it pushed itself, and thus \(r'\ange{z}\) is valid.
                    If \(\stk(r\ang{\lambda^{n}}) = \stk(r\ange{v})\) and \(\stk(r\ange{u}) = \stk(r\ang{\mu^{n'}})\), then while reading \(z\), \(r\) can only pop elements that were pushed by \(u\).
                    Because the definition of \(A'\) imposes in that case that \(\lambda^{i}\#\mu^{j}\) leaves the stack intact in \(s\), \(r'\ange{z}\) is valid.
                    If \(\stk(r\ang{\lambda^{n}}) = \stk(r\ange{v})\) and \(\stk(r\ange{u}) \neq \stk(r\ang{\mu^{n'}})\), then \(r\ang{\lambda^{n}}\) or \(r\ang{\mu^{n'}}\) pushes an element on the stack that is not popped.
                    Because none of \(\lambda^{n}\) and \(\mu^{n'}\) are \(r\)-coupled with \(z\), \(z\) only pops elements that it pushed itself, and thus \(r'\ange{z}\) is valid.
                    In all four cases, the stack semantic of \(r'\ang{u\lambda^{i}v\mu^{j}z}\) is correct.
            \end{itemize}
            We have shown that \(u\lambda^{i}v\mu^{j}z\in L\).

            It remains to show that \(\lambda^{n}\# \mu^{n'}\) is in \(L'\).
            The reasons are similar as before, depending of the definition of \(A'\), thus we only sketch them.
            Let \(r'\) be the run in \(A'\) equal to \(r\ang{\lambda^{n}}er\ang{\mu^{n'}}\).
            First, \(u\) and \(\lambda^{n}\) are not \(r\)-coupled and thus \(\lambda^{n}\) does not pop anything pushed by \(u\).
            Thus the stack semantic of \(\lambda^{n}\) in \(A'\) is valid.
            If \(\stk(r\ange{u}) = \stk(r\ang{\lambda^{n}})\), then \(\stk(r'\ang{\lambda^{n}}) = \emptyset\) and thus the transition \(e\) can be used.
            If \(\stk(r\ang{\lambda^{n}}) \neq \stk(r\ange{v})\), then, by the non \(r\)-coupling, \(r\ang{\mu^{n'}}\) does not pop anything pushed by \(r\ang{\lambda^{n}}\) and the stack can correctly be emptied.
            If \(\stk(r\ang{\lambda^{n}}) \neq \stk(r\ange{v})\) and \(\stk(r\ange{v}) = \stk(r\ang{\mu^{n'}})\), then, by the non \(r\)-coupling, \(r\ang{\mu^{n'}}\) does not pop anything pushed before and reach \(p_{4}\) with an empty stack.
            If \(\stk(r\ang{\lambda^{n}}) = \stk(r\ange{v})\) and \(\stk(r\ange{u}) = \stk(r\ang{\mu^{n'}})\), then \(r\ang{\mu^{n'}}\) pops everything that was pushed by \(r\ang{\lambda^{n}}\) and thus reaches \(p_{4}\) with an empty stack.
  \end{proof}

  Intersection with a regular language preserves CFLness, thus  \(L'\cap  \lambda'\lambda^{*}\#\mu^{*} \mu'\) is again a CFL and Ogden's lemma can be applied.
  The bound \(p\) in the statement of the lemma is uniform in all runs, therefore its definition must not depend on the precise run.
  We notice that there are only finitely many possible PDA \(A'\): it depends only on \(p_{1},p_{2},p_{3},p_{4}\), the prefix and suffix, and the four stack conditions.
  In turn, for all possible accepting runs, there are finitely many possible  \(L'\cap  \lambda'\lambda^{*}\#\mu^{*} \mu'\).
  Let \(p\) be the maximum of the pumping constants given by~\cref{lem:unbalanced_ogden} for these languages; it does not depend on \(r\).
  By~\cref{claim:subclaim_pumping2}, \(\lambda'\lambda^{n}\#\mu^{n'}\mu'\in L'\cap  \lambda'\lambda^{*}\#\mu^{*} \mu'\).
  Applying~\cref{lem:unbalanced_ogden}, there are \(0\leq l<k\leq p\) such that for all \(m\geq -1\), \(\lambda'\lambda^{n+km}\#\mu^{n'+lm}\mu'\in L'\).
  For any \(m\geq -1\), the other part of~\cref{claim:subclaim_pumping2} gives that \(u\lambda'\lambda^{n+km}v\mu^{n'+lm}\mu'z\in L\).
\end{proof}

\LemFindingSpots*
\begin{proof}
  We only prove the first conclusion, the second being symmetric.
  Write the word \(ux_{1}\cdots x_{n}vy_{n}\cdots y_{1}z\) where every \(x_{i}=x\) and every \(y_{i}=y\).
  First, notice that every \(x_{i}\) can be decomposed into five parts of size \(\pc(\lambda,\mu)\cdot |\alpha\beta\gamma|\cdot|A|\cdot |\lambda|\).
  We prove that if any part of \(x_{i}\) is not \(r\)-coupled with \(y^{n}\) at all, then the pumping condition holds.

  For brevity, we say a given letter \(w[i]\) is in an \(r\)-coupling if it's
  transition \(r\langle i \rangle\) is. Note, that \(u\),\(v\) and \(w\) can
  only be coupled with at most \(|\alpha\beta\gamma|\cdot |A|\) letters of the
  part of \(x_{i}\).  By the definition of \(p\), there is a contiguous portion
  of the part of \(x_{i}\) with at least \(\pc(\lambda,\mu)\) repetitions of
  \(\lambda\) that is not \(r\)-coupled with any letter; though they may be
  coupled with a loop of \(\eps\)-transitions, these would notably be pumpable.
  Thus we can apply~\cref{lem:pumping2}, to have \(k\leq \pc(\lambda,\mu)\) such
  that \(ux^{n}\lambda^{km}vy^{n}z\) is in \(L\) for all \(m\geq -1\).  This
  gives the pumping condition by allocating furthermore \(l=0\).

  Let \(i\) such that \(x_{i}\) is not \(r\)-coupled with \(y_{i}\).
  Thus we have that the middle part of \(x_{i}\) is \(r\)-coupled with \(y^{n}\); say with \(y_{j}\) starting at position \(s\) in \(x_{i}\) and ending at position \(t\) in \(y^{n}\).
  We assume that \(j<i\), the other case being symmetric.
  We know that the first part of \(x_{1}\) is \(r\)-coupled with \(y^{n}\) as well; say from position \(s'\) in \(x_{1}\) and ending at position \(t'\) in \(y^{n}\).
  We decompose \(w\) into \(w_{1}w_{2}w_{3}w_{4}w_{5}\) with \(w_{1} = w[:s'[\), \(w_{2} = w[s':s]\), \(w_{3}=]s:t[\), \(w_{4}=[t:t']\) and \(w_{5}=]t':]\).
  It is impossible for \(w_{2}\) and \(w_{4}\) to be \(r\)-coupled with \(w_{1},w_{3}\) and \(w_{5}\), unless there is a crossing of couplings.
  Moreover, on one hand, \(w_{2}\) contains at least all of \(x_{2}\cdots x_{i-1}\), plus four parts out of five of \(x_{1}\), plus two parts of \(x_{i}\), giving \(|w_{2}|/|\lambda|\geq (i-1)p + \pc(\lambda,\mu)\cdot |\alpha\beta\gamma|\cdot|A|\cdot |\lambda|\)
  On the other hand, \(w_{4}\) is included in \(y_{i-1}\cdots y_{1}\), and thus \(|w_{4}|/|\mu|\leq (i-1)p\).
  Wrapping it all together, we obtain \(|w_{2}|/|\lambda|\geq |w_{4}|/|\mu| + \pc(\lambda,\mu)>|w_{4}|+1\).
  Finally, \(w_{2}\) can be written as a power of a circular permutation \(\lambda\) followed by a prefix of it; and \(w_{4}\) can be written as a power of a circular permutation \(\mu\) preceded by a suffix of it.
  We can therefore apply~\cref{lem:pumping2} on this decomposition to obtain the desired pumping.
\end{proof}

\section{Proofs of \Cref{sec:shuf}}
\PropDecideAntigrids*
\begin{proof}
  Since semilinear sets are effectively closed under complement, it is sufficient to argue that we can find grids in them.
  Let \(\phi\) be a formula defining a semilinear set \(S\).

  Without loss of generality, we assume that all lines involved in \(\phi\) are non-intersecting.
  To do that, we notice that we can shift the origin of the plane without affecting the existence of an antigrid.
  Formally, there are only finitely many lines in \(\phi\), so take \((i,j)\) be a point such that every intersection point \((i',j')\) occurs such that \(i'<i\) or \(j'<j\).
  Then define \(S' = (S - (i,j)) \cap \bbN^{2}\); it is again semilinear by replacing every occurrence of \(x\) (resp. \(y\)) in \(\phi\) by \(x-i\) (resp. \(y-j\)).
  We can drop every half-plane atomic formula so that the remaining ones have non-intersecting lines.
  It is clear that there is an antigrid is \(S\) iff there is one is \(S'\).

  For \(l_{1}\) and \(l_{2}\) two non-intersecting lines, we can define an order that corresponds to the intuition that \(l_{1}\) is on the left of \(l_{2}\).
  To do that, notice that \(l_{1}\) separates the plane \(\bbN^{2}\) in two regions: one with infinitely many points of the x-axis, and the other with infinitely many points of the y-axis.
  Then \(l_{1}\leq l_{2}\) iff \(l_2\) belongs to the latter region.
  This allows one to enumerate all lines in \(\phi\) from left-to-right, according to that order: \(l_{1}\leq l_{2}\leq \cdots \leq l_{n}\).
  We consider the polyhedra \(p_{0},\ldots, p_{n}\) such that \(p_{i}\) is the region strictly between \(l_{i}\) and \(l_{i+1}\).
  Thus the \(l_{i}\) and \(p_{i}\) form a partition of \(\bbN^{2}\).
  We sort them into cones and bands.
  Any region between two non-parallel lines is called a \emph{cone}.
  Any region between two parallel lines is called a \emph{band}.
  We also consider that the lines are bands.
  We call any of these two a \emph{region}.
  The partition now reads: \(\bbN^{2} = \bigcup_{r \text{ region}} r = \bigcup_{c \text{ cone}}c \cup \bigcup_{b \text{ band}} b\).

  We can now take into account the periodic information contained in \(S\).
  Here we assume that every modular atomic predicate uses the same modulus \(p\), this can be ensured by taking the least common multiple of all moduli and adapting constants.
  Assume we only consider points in a fixed region \(r\).
  Then every half-plan atomic formula in \(\phi\) evaluates to a constant, and thus \(\phi\) is reduced to a Boolean combination of modular atomic formula.
  Thus adding \((p,0)\) or \((0,p)\) to any point does not affect membership in \(S\cap r\).
  Let \(\core(r)\) be the subset of \([0,p-1]^{2}\) such that any \((i,j)\) is \(r\) is in \(S\) iff \((i \text{ mod } p, j \text{ mod } p)\in \core(r)\).
  This gives a partition of \(S\): \(S=\bigcup_{r\text{ region}} r \cap (\core(r) + (p\bbN)^{2})\).

  The rest of the proof is devoted to proving the desired condition: \(S\) has a grid iff \(\bigcap_{c \text{ cone}} \core(c) \neq \emptyset\).
  First, assume that \(S\) has a grid \(U\times V\).
  Without loss of generality, by taking a subgrid, we can assume that \(U=i + kp\bbN\) and \(V = j+kp\bbN\), for some \(i,j,k\).
  The key observation about cones, that is not shared by bands, is that the grid intersects each one of them.
  Indeed, the grid intersects any square of length \(kp\), and such a square can be found is any cone \(c\).
  Let \(c\) be a cone, and \((i,j)\in U\times V \cap c\).
  By definition of the core, this implies that \((i \text{ mod } p, j \text{ mod } p)\in \core(c)\).
  However, because of the period \(kp\) of the grid, any two points on the grid have the same image modulo \(p\).
  Thus we have found a common point to the core of every cone.

  Second, assume that \(\bigcap_{c \text{ cone}} \core(c) \neq \emptyset\) and let \((i,j)\in[0,p-1]^{2}\) be a common point.
  Consider the grid \(G = (i,j) + (p\bbN)^{2}\).
  On any cone, \(G\) is included in \(S\), but might include points not in \(S\) in some bands.
  Let \(b\) be a band, we will identify a subgrid of \(G\) that avoids \(b\).
  Any band can be written as a finite union of lines, hence we only show the case where \(b\) is a single line of equation \(cx + dy + e =0\).
  If the intersection of \(G\) and \(b\) is empty, we are done; otherwise it is infinite and we can pick two points \((i',j'), (i'+kp,j'+lp)\in G\cap b\).
  Such two points on a line implies that adding any multiple of \((kp,lp)\) to a point preserves membership in \(b\).
  Consider the subgrid \(G' = (i'+kp,j') + (2kp\bbN \times 2lp\bbN)\).
  In other words, we pick a point in \(G\) close to \(b\), and we start a coarser subgrid from there.
  It remains to show that \(G'\) does not intersect \(b\).
  By contradiction, assume there is a point \((i'+(2u+1)kp, j' + 2vlp)\in b\).
  Because there is a unique point on a line with a given y-axis, and \( (i'+2vkp,j'+2vlp)\in b\), we deduce that \( i'+(2u+1)kp = i'+2vkp\) which in turn gives \(2u+1=2v\).
  This is impossible.
  Finally, we can iterate the process to avoid every band and obtain a grid included in \(S\).
\end{proof}

To expand on \Cref{rem:complexity}, the previous decidability algorithm
yields a \(\PSPACE\) upper bound. This bound stands for finding both grids and
antigrids, as complementing a formula does not yield a blowup (and the
algorithm works also with CNFs). We start by describing the size of a formula
\(\psi\): any atomic formula with parameters \(c,d,e\) (and \(p\)) has size
\(|c|+|d|+|e|+ |p|\), and \(|\psi|\) is the sum of the size of all its atomic
subformulas.  The first step of the algorithm is to shift the formula to avoid
intersecting lines.  By Cramer's rule, the intersection point of two lines
only incur a quadratic blowup.  Thus the new shifted formula only has a
polynomial blowup.  Moreover, there is a linear number of lines involved, and
thus a linear number of regions.  The common period \(p\) can be, on the worst
case where every involved periods are different primes, exponentially large.
This implies that the size of the cores, and their computation, take an
exponential time.  The final condition is simply the computation of the
intersection of  linear number of exponential sets, yielding an \(\EXPTIME\)
algorithm.  This can be strengthened into \(\coPSPACE\) by noting that we can
guess an element and check if it is in the intersection.  The checking takes
polynomial time as arithmetic operation are polynomial in the bitsize, and
thus checking membership in regions and in the semilinear set can be done in
polynomial time.  We conclude by Savitch's theorem.  Note that the main blowup
comes from the size of the global period, therefore whenever the LCM of all
periods is polynomially related to the size of \(\psi\) (for instance when
they are all the same), we have a \(\PTIME\) algorithm instead.
}{}
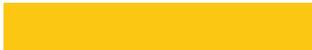
\end{document}